\documentclass{article}
\usepackage{graphicx} 

\usepackage[english]{babel}
\usepackage[utf8]{inputenc}
\usepackage[T1]{fontenc}
\usepackage{lmodern}
\usepackage[a4paper, margin=1in]{geometry}
\usepackage{graphicx}
\usepackage{framed}
\usepackage[framemethod=tikz]{mdframed}
\usepackage{xcolor}
\usepackage[most]{tcolorbox}
\usepackage{todonotes}
\usepackage{color}

\definecolor{darkgreen}{rgb}{0,0.5,0}
\definecolor{darkgray}{rgb}{0.2,0.2,0.2}
\usepackage{hyperref}
\hypersetup{
    unicode=false,          
    colorlinks=true,        
    linkcolor=blue,         
    citecolor=purple,       
    filecolor=magenta,      
    urlcolor=cyan           
}

\usepackage{amsthm}
\usepackage{amsmath}
\usepackage{amssymb}
\usepackage{amsfonts}
\usepackage{mathrsfs}
\usepackage{mathtools}
\usepackage{verbatim}
\usepackage{footnote}

\usepackage{algorithm}
\usepackage{algorithmicx}
\usepackage[noend]{algpseudocode}

\usepackage{lineno}
\usepackage{caption}
\usepackage{framed}
\usepackage{enumerate}
\usepackage{hyperref}

\usepackage{tikzsymbols}
\usepackage{thmtools,thm-restate}
\usepackage{nicefrac}

\usepackage{subcaption}
\usepackage{makecell}
\usepackage{pdfpages}
\usepackage{multirow}
\usepackage{tabularx}
\usepackage{float}
\usepackage{booktabs}

\usepackage[capitalize, nameinlink]{cleveref}
\usepackage[section]{placeins}
\crefname{theorem}{Theorem}{Theorems}
\Crefname{lemma}{Lemma}{Lemmas}
\Crefname{invariant}{Invariant}{Invariants}
\Crefname{claim}{Claim}{Claims}
\Crefname{observation}{Observation}{Observations}
\Crefname{@algorithm}{Algorithm}{Algorithms}
\Crefname{figure}{Figure}{Figures}
\crefname{appendix}{Appendix}{Appendices}
\Crefname{appendix}{Appendix}{Appendices}

\newtheorem{theorem}{Theorem}[section]
\newtheorem{lemma}[theorem]{Lemma}
\newtheorem{corollary}[theorem]{Corollary}
\newtheorem{definition}[theorem]{Definition}

\newtheorem{observation}[theorem]{Observation}
\newtheorem{claim}[theorem]{Claim}
\newtheorem{remark}{Remark}
\newtheorem*{remark*}{Remark}

\newcounter{resultctr}
\renewcommand{\theresultctr}{\arabic{resultctr}}

\newtcolorbox{resultbox}[1][]{%
  enhanced,
  colback=gray!10,
  colframe=gray!60,
  coltext=black,
  boxrule=0.6pt,
  arc=2mm,
  left=1.2mm,
  right=1.2mm,
  top=1.0mm,
  bottom=1.0mm,
  before title={\refstepcounter{resultctr}}, 
  title={Result~\theresultctr\if\relax\detokenize{#1}\relax\else:~(#1)\fi},
  fonttitle=\bfseries,
  coltitle=black,
  attach title to upper,
  after title=\par\smallskip,
}

\DeclareMathOperator{\poly}{poly}

\def\LOCAL{\ensuremath{\mathsf{LOCAL}}\xspace}

\newcommand{\prob}[1]{\Pr \left[ #1 \right]}
\newcommand{\eqdef}{\stackrel{\text{\tiny\rm def}}{=}}

\newcommand{\E}[1]{\mathbb{E}\left[#1\right]}
\newcommand{\tO}{\tilde{O}}
\newcommand{\rb}[1]{\left( #1 \right)}

\newcommand{\cP}{\mathcal{P}}

\newcommand{\eps}{\epsilon}
\newcommand*{\defeq}{\stackrel{\text{def}}{=}}
\newcommand{\LocalPrune}{\textsc{Local-Prune}\xspace}
\newcommand{\NumPathsIn}{\mathsf{NumPathsIn}}
\newcommand{\NumPathsOut}{\mathsf{NumPathsOut}}
\newcommand{\Expo}{\textsc{Exponentiate-and-Prune}\xspace}
\newcommand{\PartialLayer}{\textsc{Partial-Layer-Assignment}\xspace}
\newcommand{\PartialLayerDS}{\textsc{Partial-Layer-Assignment-DS}\xspace}
\newcommand{\PartialLayerTree}{\textsc{Partial-Layer-Tree-Assignment}\xspace}
\newcommand{\kcorePeeling}{\textsc{K-Core-Peeling}\xspace}
\newcommand{\dsPeeling}{\textsc{DS-Peeling}\xspace}
\newcommand{\kcoreSimulation}{\textsc{K-Core-Simulation}\xspace}
\newcommand{\dsSqrtPeeling}{\textsc{DS-Sqrt-Peeling}\xspace}
\newcommand{\Missing}{\mathsf{Missing}}
\newcommand{\map}{\mathsf{map}}
\newcommand{\children}{\mathsf{children}}

\newcommand{\cAMPC}{\mathcal{A}_{\textsc{MPC}}}

\newcommand{\LazyPathsIn}{\NumPathsIn_{G,\ell_G,L}}

\newcommand{\pruned}{\text{pruned}}
\newcommand{\polyloglog}{\poly\log\log}

\title{Fixed-Threshold Peeling in Sublinear MPC: Round-Approximation Tradeoffs and Applications}
\author{Slobodan Mitrović\thanks{The authors are supported by NSF Faculty Early Career Development Program No.~2340048. e-mail: \texttt{\{smitrovic, thjpan, wsheu\}@ucdavis.edu}} \\ UC Davis
\and Theodore Pan\footnotemark[1] \\ UC Davis
\and Wen-Horng Sheu\footnotemark[1] \\ UC Davis}

\date{ }

\begin{document}

\maketitle

\begin{abstract}
A number of fundamental graph problems admit simple algorithms based on iterative peeling: repeatedly remove all vertices whose current degree is below a fixed threshold. 
This paradigm underlies algorithms for density-dependent edge orientation, density-dependent coloring, densest subgraph, and $k$-core decomposition. 
In this paper, we study these problems in the sub-linear MPC model.

We achieve the following round-approximation tradeoffs for this family of problems.
For any constant $\delta \in (0,1)$, our algorithms use
$O(n^\delta)$ local memory and $\widetilde O(m+n^{1+\delta})$ global memory.

\begin{itemize}
    \item \textbf{Density-dependent edge orientation.}
    For any integer $t > 0$, we compute an orientation with maximum out-degree at most
    $(2+\epsilon)(t+1)\alpha(G)$ in $O(\lg^{1/(t+2)} n \cdot \operatorname{poly}(\lg \lg n))$ rounds, where $\alpha(G)$ denotes the minimum possible maximum out-degree of an orientation of $G$. 
    In the $\operatorname{poly}(\lg\lg n)$-round regime, obtained by setting $t=\lceil \lg\lg n/\lg\lg\lg n\rceil$, this gives an $O(\lg\lg n/\lg\lg\lg n)$-approximation, improving the approximation factor of the recent work by Ghaffari and Grunau~[PODC 2025].

    \item \textbf{Densest subgraph.}
    We obtain a $(4+\epsilon)$-approximation in $\widetilde O(\lg^{1/3} n)$ MPC rounds and a $(6+\epsilon)$-approximation in $\widetilde O(\lg^{1/4} n)$ MPC rounds.
    This improves the $\widetilde O(\sqrt{\lg n})$ round complexity of Ghaffari, Lattanzi, and Mitrović~[ICML 2019] with a slightly larger approximation factor.
    This is the first $O(1)$-approximate algorithm for densest subgraph to break the $\Theta(\sqrt{\lg n})$ round-complexity barrier in the sub-linear MPC model.

    \item \textbf{$k$-core decomposition.}
    For any integer $t > 0$, we compute approximate coreness values within a factor of $(2+\epsilon)(t+1)$ in
    $O(\lg^{1/(t+2)} n \cdot \operatorname{poly}(\lg \lg n))$ MPC rounds.
    This improves the $\widetilde O(\sqrt{\lg n})$ round complexity of Ghaffari, Lattanzi, and Mitrović~[ICML 2019], again giving a round-approximation tradeoff.

    \item \textbf{Density-dependent coloring.}
    By combining our orientation algorithm with standard coloring reductions, we obtain a coloring with $(2+\epsilon)(t+1)\alpha(G)+1$ colors in $O(\lg^{1/(t+2)} n \cdot \operatorname{poly}(\lg \lg n))$ MPC rounds. 
    In particular, this improves the approximation obtainable in $\operatorname{poly}(\lg\lg n)$ rounds over Ghaffari and Grunau~[PODC 2025].
\end{itemize}
\end{abstract}
\newpage

\tableofcontents
\newpage

\section{Introduction}
The Massively Parallel Computation (MPC) model was defined in a series of works~\cite{dean2008mapreduce,karloff2010model,goodrich2011sorting,andoni2014parallel}, and by now has become a standard abstraction for modern large-scale computation frameworks, such as MapReduce, Hadoop, Spark, Dryad, and Flume.
In this model, the system consists of a number of machines, each with local memory $S$. 
The machines perform computation in synchronous rounds. Between two consecutive rounds, machines exchange messages: each machine can communicate with any other machine, subject to the constraint that the total size of messages sent and received by each machine is at most $S$.
Early results established that, under standard memory assumptions, one PRAM round can be simulated in $O(1)$ MPC rounds, implying that the vast literature of algorithms in PRAM can be almost directly ported to MPC with only constant overhead in the round complexity.
Thus, direct PRAM simulation gives a general baseline for MPC algorithms, but a central goal in the area has been to design MPC algorithms with substantially smaller round complexity than this baseline.
The extent to which such speed-ups are possible depends strongly on the local memory $S$, leading to three standard regimes.
For graph problems on $n$ vertices and $\delta \in (0, 1)$, those regimes are: the super-linear regime, where $S = n^{1+\delta}$; the near-linear regime, where $S = n \cdot \poly(\log n)$; and the sub-linear regime, where $S = n^\delta$.

Lattanzi, Moseley, Suri, and Vassilvitskii~\cite{lattanzi2011filtering} demonstrated that a number of fundamental problems can be solved in only $O(1)$ rounds in the super-linear MPC regime.
These problems include computing connected components, minimum spanning tree, maximal matching, approximate maximum weighted matching, edge covers, and min cut, for which the state-of-the-art PRAM algorithms require at least $O(\log n)$ rounds.
Thus, the super-linear regime admits dramatic round-complexity improvements over PRAM-style bounds, but these techniques do not yield analogous speed-ups once the local memory is reduced to near-linear.

Andoni, Nikolov, Onak, and Yaroslavtsev~\cite{andoni2014parallel} initiated the study of such speed-ups in the near-linear regime.
They showed that certain geometric graph problems, e.g., minimum spanning tree and earth-mover distance, can be approximated in $O(1)$ MPC rounds.
A few years later, Czumaj, Łącki, Mądry, Mitrović, Onak, and Sankowski~\cite{czumaj2018round} introduced a round-compression technique, that enabled them to solve approximate maximum matching in $\poly(\log \log n)$ rounds in the near-linear regime, using exponentially fewer MPC rounds than what follows from direct PRAM simulation.
Variants of this approach have since been applied to several other problems, including vertex cover~\cite{Assadi17,assadi2019coresets,ghaffari2020massively}, maximal matching~\cite{behnezhad2019exponentially}, $b$-matching~\cite{ghaffari2022massively}, correlation clustering~\cite{cambus2021massively}, maximal independent set~\cite{ghaffari2018improved}, and set cover~\cite{dhulipala2024parallel}.
Thus, in the near-linear regime, round compression has led to a broad class of algorithms whose MPC round complexity is exponentially smaller than what is obtained by direct PRAM simulation. 

The situation is much less understood in the most restrictive, fully scalable sub-linear memory regime.
In this regime, the local memory is too small to directly implement the round-compression techniques that have been successful in the near-linear setting.
Ghaffari and Uitto~\cite{ghaffari2019sparsifying} introduced a combination of graph sparsification and exponentiation that improves the round complexity from $O(\log n)$ to $\tO(\sqrt{\log n})$ for maximal independent set, maximal matching, and $2$-approximate vertex cover.
Follow-up works obtained qualitatively similar improvements, namely roughly $\sqrt{\log n}$-factor speed-ups, for densest subgraph, graph orientation, $k$-core clustering, and set cover~\cite{ghaffari2019improved,dhulipala2024parallel,mitrovic2026new}, while the authors of \cite{chang2019complexity} obtained a much faster algorithm for $(\Delta+1)$-list coloring.
Recently, Ghaffari and Grunau~\cite{ghaffari2025density} showed that $O(\log \log n)$ approximate graph orientation can be computed in only $\poly(\log \log n)$ MPC rounds, exponentially faster than previously known. A direct application of this result also implies a faster algorithm for density-dependent graph coloring.
However, beyond these specific settings, it remains unclear which structural properties of graph problems can be exploited to obtain substantially faster algorithms in the sub-linear MPC regime.
This leaves open the following central algorithmic question:
\begin{center}
    \emph{What techniques yield even faster algorithms in the sub-linear MPC regime?}
\end{center}
We make progress on this question for fixed-threshold peeling-based graph problems, obtaining faster sub-linear MPC algorithms and, for some of these problems, exponential improvements over the previously known round complexities.

\subsection{Our results}

\label{sec:results}
\begin{table}[]
    \renewcommand{\arraystretch}{1.2}
    \centering
    \begin{tabular}{|l|l|l|l|}
    \hline
    Reference & Approximation & Round complexity & Global memory\\
    \hline
    
    \multicolumn{4}{|l|}{\textbf{Density-dependent edge orientation}} \\
    \hline
    \cite{ghaffari2019improved} & $(2+\eps)$ & $\tO(\sqrt{\log n})$  & $\tO(\lambda n)$ \\
    \hline
    \cite{ghaffari2025density} & $O(\log \log n)$ & $O(\poly(\log \log n))$  & $\tO(n + m)$ \\
    \hline
    Our work & $(2+\eps)\cdot(t+1)$ for $t \in \mathbb{Z}_{\geq 0}$ & $O\rb{\log^{1/(t+2)} (n) \cdot \polyloglog(n)}$ & $\tO(n + m)$ \\
    \hline
    
    \multicolumn{4}{|l|}{\textbf{Density-dependent coloring}} \\
    \hline
    \cite{ghaffari2025density} & $O(\log \log n)$ & $O(\poly(\log \log n))$  & $\tO(n + m)$ \\
    \hline
    Our work & $(2+\eps)\cdot(t+1)$ for $t \in \mathbb{Z}_{\geq 0}$ & $O\rb{\log^{1/(t+2)} (n) \cdot \polyloglog(n)}$ & $\tO(n + m)$ \\
    \hline
    
    \multicolumn{4}{|l|}{\textbf{Densest subgraph}} \\
    \hline
    \cite{ghaffari2019improved} & $(1+\eps)$ & $\tO(\sqrt{\log n})$  & $\tO(n^{1+\delta} + m)$ \\
    \hline
    Our work & $(4+\eps)$ & $\tO\rb{\log^{1/3} n}$  & $\tO(n^{1+\delta} + m)$ \\
    \hline
    Our work & $(6+\eps)$ & $\tO\rb{\log^{1/4} n}$  & $\tO(n^{1+\delta} + m)$ \\
    \hline
    
    \multicolumn{4}{|l|}{\textbf{$k$-core decomposition}} \\
    \hline
    \cite{ghaffari2019improved} & $(2+\eps)$ & $\tO(\sqrt{\log n})$  & $\tO(n^{1+\delta} + m)$\\
    \hline
    Our work & $(2+\eps)\cdot(t+1)$ for $t \in \mathbb{Z}_{\geq 0}$ & $O\rb{\log^{1/(t+2)} (n) \cdot \polyloglog(n)}$  & $\tO(n^{1+\delta} + m)$\\
    \hline

    \end{tabular}
    \caption{A summary of our results.}
    \label{tab:summary}
\end{table}
\


Our main result is obtained by improving the exponentiate-and-prune algorithm of \cite{ghaffari2025density}.
In particular, our technique achieves an improved approximation factor in $\poly\log\log(n)$ rounds for edge orientation and density-dependent coloring.
Also, it yields new tradeoffs between the approximation factor and round complexity.
Additionally, our approach can be extended to solve the densest subgraph and $k$-core decomposition problems.
All these problems admit fixed-degree threshold peeling algorithms that attain constant approximation.
Consequently, our technique can be viewed as a general approach for simulating such algorithms.
We refer the reader to \cref{sec:basic-definition} for formal definitions of these problems.
\cref{tab:summary} provides a summary of our results.

Our new tradeoff leads to several improvements on the state-of-the-art.
Given any non-negative integer $t$, our algorithms compute a $(2+\eps) \cdot (t+1)$-approximation in $O(\log^{1/(t+2)}(n) \cdot \polyloglog(n))$ rounds (except for the densest subgraph problem, where we require $t \leq 2$).
At one extreme, when $t$ is a constant, we obtain the first $o(\sqrt{\log n})$-round constant-approximation algorithms in the sublinear MPC model for the problems of edge orientation, density-dependent coloring, densest subgraph, and $k$-core decomposition.

\begin{resultbox}
There exist $o(\sqrt{\log n})$-round $O(1)$-approximate algorithms in the MPC sub-linear memory regime for the following problems: edge orientation, density-dependent coloring, densest subgraph, and $k$-core decomposition.
For constant $t \geq 1$, the algorithms compute a $(2+\eps)\cdot(t+1)$-approximation using $\tO\rb{\log^{1/(t+2)} n}$ rounds.
    For densest subgraph, we require an additional constraint of $t\leq 2$.
\end{resultbox}

At the other extreme, when $t = \lceil\log\log n/\log \log \log n\rceil$, our result gives $O(\poly\log \log n)$-round $O(\log\log n/\log \log\log n)$-approximation algorithms.
This improves the best-known approximation factor for $O(\poly\log \log n)$-round algorithms by a $\Theta(\log\log\log n)$ factor.
\begin{resultbox}
There exist $O(\poly\log \log n)$-round $O(\log\log n/\log \log\log n)$-approximate algorithms in the MPC sub-linear memory regime for the following problems: edge orientation, density-dependent coloring, and $k$-core decomposition.
\end{resultbox}

\noindent \textbf{Remark on global memory.}
The global memory complexity of our algorithms matches the state-of-the-art bounds in \cite{ghaffari2019improved,ghaffari2025density}:
Our edge orientation and density-dependent coloring algorithms achieve the optimal $\tO(m + n)$ global memory, while the densest subgraph and $k$-core decomposition algorithms require $\tO(m + n^{1+\delta})$ global space.
This distinction is due to additional technical challenges posed by the latter two problems, making it difficult to apply a previous technique called \emph{budget boosting} for reducing memory usage.
We provide more details in \cref{sec:outline}.

\subsection{An outline of our approach}
\label{sec:outline}

Our starting point is inspired by the recent result of Ghaffari and Grunau~\cite{ghaffari2025density} for density-dependent edge orientation in the sub-linear MPC regime.
We first recall the high-level idea behind their approach, as it will also serve as the point of comparison for our algorithms.

\subsubsection{Outline of \cite{ghaffari2025density}}
The work of Ghaffari and Grunau~\cite{ghaffari2025density} studies approximate edge orientation. In this problem, the goal is to orient the edges of an input graph so as to minimize the maximum outdegree. Equivalently, one seeks an orientation whose maximum outdegree is close to the arboricity $\lambda$, or to the density of the densest subgraph, of the input graph.

The classical peeling algorithm for this problem repeatedly performs the following step. 
Let $S$ be the set of vertices whose degree in the current graph is $O(\lambda)$. 
For each vertex $v \in S$, orient all edges incident to $v$ away from $v$, and then remove $v$ from the graph. We say that $v$ is peeled in this iteration.
It is not difficult to show that this process gives an $O(1)$-approximate orientation and terminates after $O(\log n)$ peeling iterations.

Thus, the peeling iteration in which a vertex $v$ is removed is determined by the $O(\log n)$-hop neighborhood of $v$.
We represent the peeling times by a labeling function $\ell: V \rightarrow \{1, 2, \dots, O(\log n)\}$ such that $\ell(v)$ is the iteration in which $v$ is peeled.
Such a labeling function is called a \emph{(partial) layer assignment}.
If these labels were known for all vertices, then each edge $\{u,v\}$ could be oriented from the endpoint peeled earlier to the endpoint peeled later, with ties broken arbitrarily.

This suggests a natural MPC strategy. For each vertex $v$, gather its $O(\log n)$-hop neighborhood, simulate the peeling process locally inside this neighborhood, and thereby determine the label of $v$. If memory were not a constraint, such neighborhoods could be gathered in only $O(\log \log n)$ MPC rounds by graph exponentiation. The obstruction is that some $O(\log n)$-hop neighborhoods may be too large to fit in the local memory of a single machine.

The main idea of Ghaffari and Grunau is to perform graph exponentiation on \emph{pruned} neighborhood views. 
That is, instead of gathering the full $O(\log n)$-hop neighborhood of a vertex, the algorithm constructs a reduced local tree-like view during the exponentiation process. 
Whenever such a view becomes too large, the algorithm prunes $O(\lambda)$ of its heaviest subtrees. 
The edges corresponding to these pruned parts are no longer used in the local simulation; instead, they are charged directly to the approximation guarantee and can be oriented later arbitrarily without affecting the correctness of the simulated peeling step. 
In each exponentiation iteration, each vertex is charged for only $O(\lambda)$ such discarded edges. Since the exponentiation process has only $O(\log \log n)$ iterations, the total additional outdegree incurred by these discarded edges is $O(\lambda \cdot \log \log n)$.

Using this idea, \cite{ghaffari2025density} showed that the pruned $T$-hop neighborhoods of \emph{most} vertices have $\tO(\lambda^T)$ size.
However, when $\lambda$ is large, the pruned neighborhood views can still be too large to fit into the local memory of one machine.
To address this issue, they proposed a preprocessing procedure, which partitions the edge set into subsets of $\Theta(\log n)$ arboricity.
An approximate edge orientation is then computed for each edge subset independently and in parallel.
The resulting orientations are combined to obtain the final output.

This preprocessing allows them to assume $\lambda = \Theta(\log n)$, thereby reducing the problem to uniformly sparse graphs.
Consequently, setting $T = \Theta(\delta\log n/\log \log n)$ yields an $\tO(\lambda^T) = \tO(n^\delta)$ upper bound on the size of the $T$-hop neighborhood, which is sufficiently small to fit on a single machine in the sublinear MPC setting.
Importantly, this size upper bound only holds for most vertices; some vertices may still have neighborhoods that are too large even after pruning.
These vertices are marked as \emph{inactive}. 
It can be shown that the number of inactive vertices is small enough to be ignored.


\subsubsection{Improved exponentiate-and-prune algorithm}
The main source of the super-constant approximation factor in the approach above is the pruning step. 
Thus, reducing the number of pruning steps is a natural way to improve the approximation factor.
However, the challenge is that pruning is also what keeps the local tree views small enough to fit within the memory of one machine. 
If one simply performs more graph exponentiation steps without pruning, the corresponding tree views may become too large. 
Our approach is to balance these two effects. Instead of pruning after every exponentiation step, we first keep the tree views sufficiently below the local memory budget, then perform several graph exponentiation steps without pruning, and only then apply the pruning routine. We repeat this process several times.

Our improved exponentiate-and-prune algorithm introduces the round-approximation tradeoffs we describe in \cref{sec:results} as well as a framework that allows us to attain algorithms for fixed-threshold peeling-based graph problems such as densest subgraph and $k$-core decomposition.
These problems all have a peeling algorithm that computes a constant approximation by iteratively removing all vertices with degree less than some fixed parameter $k$.
Then, we use the following structure for all of our algorithms:
\begin{enumerate}
    \item Perform a preprocessing routine on the graph so that the fixed peeling parameter $k$ is $O(\log n)$.

    \item Compute a pruned tree view for all active vertices using our exponentiate-and-prune algorithm.

    \item Compute vertex layers for each pruned tree view using the fixed-threshold peeling algorithm.

    \item Combine the vertex layers across all the pruned tree views of different vertices into a single vertex layer.

    \item Extract the final answer from the computed vertex layer.
\end{enumerate}
All the steps, except for using our exponentiate-and-prune algorithm to compute pruned tree views, are problem-specific.
Specifically, handling inactive vertices properly in the applications to densest subgraph and $k$-core decomposition is a major challenge in comparison to edge orientation.
We discuss these details next.

\subsubsection{Application to densest subgraph}
\label{sec:application-densest-subgraph}
To apply our exponentiate-and-prune algorithm to the densest subgraph problem, we first start with the fixed-threshold peeling algorithm from \cite{mitrovic2026new}.
Let $\rho^*(G)$ be the density of the densest subgraph of graph $G$.
The peeling algorithm iteratively removes all vertices with degree less than $\rho^*(G)$, ending with a $\rho^*(G)$-core which is a $2$-approximation of the densest subgraph.
However, this can require $\Omega(n)$ iterations, so \cite{mitrovic2026new} introduces a new stopping condition by comparing the sizes of vertex sets.
Specifically, if $V_i$ is the current vertex set of the graph and peeling is applied to get the new vertex set $V_{i+1}$, then the algorithm checks if $|V_{i+1}| \geq \frac{|V_i|}{1+\eps}$ for some $\eps\in(0,1)$.
In other words, the stopping condition checks if a large fraction of the vertices in $V_i$ have high degree.
If this stopping condition is satisfied, then it can be shown that $V_i$ is a $(2+\eps)$-approximate densest subgraph.
Otherwise, the condition guarantees a drop in the vertex size by a factor of at least $1+\eps$.
This gives us an algorithm using $O(\log_{1+\eps} n)$ iterations to compute a $(2+\eps)$-approximate densest subgraph.

Now, we want to simulate this peeling algorithm on the local pruned tree views produced by our exponentiate-and-prune approach.
When it comes to producing partial layer assignments, we can interpret our peeling as aiming to find the vertices in the partial layer assignment that have not been assigned a layer yet.
We say that such vertices are in the $\infty$ layer.
Unlike edge orientation which aims to compute a full layer assignment, our goal is to find the largest partial layer assignment possible and use the $\infty$ layer as our approximate densest subgraph.
Additionally, with intermediate layers, we want to try to apply our stopping condition described above to compute an approximate densest subgraph.

So, how do we produce our layer assignment from the different local prune tree views? In edge orientation, we can combine partial layer assignments from different trees into one partial layer assignment by taking the minimum layer assignment for a vertex over all the trees it appears in.
However, this does not work for densest subgraph.
Using the minimum gives us the guarantee that the out-degree of any vertex is bounded.
For densest subgraph though, we want the guarantee that the degree of any vertex that is not removed in a specific layer must have had a large degree in that layer.
Therefore, we can fix this by taking the maximum over all partial layer assignments from the different trees instead of the minimum.
We also have to be careful that vertices are not all assigned layer $\infty$ from taking the maximum, and we deal with this by leveraging the heights of the trees.
More specific details about how to take the maximum properly can be found in \cref{sec:ds:combining}.

With this layer assignment, we want to analyze the vertex set sizes by using the stopping condition of our peeling algorithm.
As described before, our goal is to find the $\infty$ layer vertices or some larger set that contains them while having a large density.
However, the $\infty$ layer vertices cause problems when we try to analyze the layer assignment.
Specifically, when the $\infty$ layer is large, it can cause the number of inactive vertices to be large as well.
We do not have any guaranteed information on whether inactive vertices should be peeled or should not be peeled, so they dilute the stopping condition we want to use from \cite{mitrovic2026new} by comparing vertex set sizes.
When the stopping condition is diluted, we no longer can guarantee that the vertex set size will drop by a constant factor after every iteration of peeling, and that could result in $\omega(\log n)$ layers.
In general, the $\infty$ layer can be just large enough to disrupt our stopping condition, but it can also be just small enough that most larger sets that contain it are not good approximations of the densest subgraph.

Then, how do we deal with the $\infty$ layer?
We introduce a subroutine that performs the peeling algorithm from \cite{mitrovic2026new} using $\tO(L/\sqrt{\log n})$ sublinear MPC rounds, where $L$ is the number of iterations of peeling we want to simulate.
The subroutine uses the framework from \cite{ghaffari2019sparsifying} that shaves off a $\sqrt{\log n}$ factor from peeling algorithms that take $O(\log n)$ iterations in the sublinear MPC model.
Normally, it would take $\tO(\sqrt{\log n})$ rounds, but by using $L = \Theta(\log^{2/3} n)$ or $L = \Theta(\log^{3/4} n)$, it takes $\tO(\log^{1/3} n)$ and $\tO(\log^{1/4} n)$ rounds, respectively. 
We are able to show that there is a perfect middle point for the size of the $\infty$ layer where if it is bigger than it, we can find an approximate densest subgraph in $\tO(L/\sqrt{\log n})$ sublinear MPC rounds using our subroutine, and if it is smaller than it, our stopping condition is not disrupted and allows us to observe a significant size reduction in the vertex set.
Therefore, we attain a MPC algorithm that uses $\tO(\log^{1/(t+2)} n)$ rounds for $t \leq 2$, where the $t\leq 2$ constraint comes from the round complexity of our subroutine.

\subsubsection{Application to $k$-core decomposition}
The starting point of our $k$-core result is a folklore peeling algorithm.
Given a degree parameter $k$, the algorithm iteratively removes the set of vertices whose current degree is smaller than $(2+\eps)k$, where $\eps > 0$ is a parameter.
It is known that the following holds after $O_\eps(\log n)$ peeling iterations.
\begin{itemize}
    \item All vertices of coreness $< k$ are peeled.
    \item No vertex with coreness $\geq (2+\eps)k$ is peeled.
\end{itemize}
A $(2+\eps)$-approximate core labeling can be obtained by running this procedure in parallel for all $k_i = (1+\eps/10)^i$, where $i = 1, 2, \dots, \log_{1+\eps/10} n$ \cite{ghaffari2019improved}.

We achieve our result by simulating this peeling algorithm.
First, we develop a preprocessing routine based on random sampling.
Consider a fixed $k$.
Let $H$ be a subgraph of $G$ obtained by sampling each edge with probability $\Theta_\eps(\frac{\log n}{k})$.
We show that running the peeling algorithm on $G$ with degree parameter $k$ can be simulated by running peeling with parameter $\Theta_\eps(\log n)$ on $H$.
(See \cref{lemma:kcore-sampling} for a more formal reduction.)
Based on this observation, we can assume without loss of generality that $k = O_\eps(\log n)$.
A similar random sampling approach appeared in \cite{ghaffari2019improved}.
However, their approach generates multiple edge samples for the simulation.
Our new analysis shows that one sample is enough for all iterations.

Our $k$-core algorithm is similar to the edge orientation algorithm.
It repeats the following process.
First, compute the pruned $L$-hop neighborhood.
Using this information, every vertex attempts to simulate $L$ iterations of peeling locally.
Each repetition is called a \emph{phase}, and can be implemented in $\polyloglog n$ MPC rounds.

The analysis for this process is challenging for the following reason.
In edge orientation, the random sampling procedure reduces both the degree parameter $k$ \emph{and} arboricity to $O(\log n)$; in addition, $k \approx 2\lambda$.
This ensures that every peeling iteration reduces the number of vertices by half.
However, this property does not hold in the $k$-core setting.
A layer can be larger than all previous ones, and furthermore, $\Theta_\eps(\log n)$ iterations of peeling may not remove all vertices.

To overcome this challenge, we strengthen the argument for edge orientation as follows.
Denote by $A_i$ the set of vertices peeled in iteration $i$.
Let $L$ be a parameter and divide the layers into \emph{batches} of $L$ layers.
Denote the $i$-th batch by $U_i = \bigcup_{j=(i-1)\cdot L+1}^{i \cdot L} A_i$.
Let $U_{\leq i}$ be $\bigcup_{j\leq i} U_j$.
Define $U_{<i}$ and $U_{> i}$ similarly.

We now outline the analysis.
By adapting the argument for edge orientation, we show that all but $O(\frac{|U_1|}{2^L})$ vertices in $U_1$ will be peeled in the first phase.
Hence, one phase suffices to reduce the number of vertices to $|U_{>1}| + \frac{|U_1|}{k^L}$.
However, in the worst case, $|U_{>1}|$ could be almost as large as $n$, and hence one phase may not even reduce the graph size by a constant factor.
In addition, it is unclear whether vertices in $U_{>1}$ can 
simulate $L$ steps of peeling in the next phase, because the simulation may depend on the remaining vertices in $U_1$, which may be outside of the $L$-hop neighborhood of $U_{>1}$.

To address this issue, we observe that some vertices in $U_{>1}$ can still simulate $L$ steps of peeling locally even though the simulation of $U_1$ is incomplete.
Intuitively, a vertex $v$ in batch $U_j$ can simulate peeling locally if there is no ``chain of dependency'' ending at $v$ that comes from $U_{<j}$.
We formulate such dependency in \cref{def:kcore-critical} and call these vertices \emph{critical vertices}.
The observation allows us to relate the size of a batch $U_j$ with the size of all previous batches, i.e., $|U_{<j}|$.
More precisely, we show that the size of $U_j$ after $i$ phases is at most $f(i, j)$, where $f(i, j)$ is defined recursively as
\[
    f(i+1, j) = f(i,j) / q^{9} + \sum_{j' < j} f(i,j') \cdot q^{j-j'+1},
\]
where we use $q \defeq k^L$.
A potential-based argument shows that $O_\eps(\log n / L)$ phases suffice to reduce the size of all batches to zero.
Hence, the simulation completes in $O_\eps(\log n / L)$ phases.
Each phase can be implemented in $O(\log\log n)$ MPC rounds.
By choosing $L$ depending on the parameter $t$, we obtain the result.

\subsection{Related work}
The most closely related work is the recent result of Ghaffari and Grunau~\cite{ghaffari2025density}, which we discuss in detail in the previous sections.
Another closely related result is due to Ghaffari, Lattanzi, and Mitrovi\'c~\cite{ghaffari2019improved}, who showed how to compute a $(1+\eps)$-approximate densest subgraph, a $(2+\eps)$-approximate $k$-core decomposition, and a $(2+\eps)$-approximate graph orientation in $\tO(\sqrt{\log n})$ rounds in the sub-linear memory regime.
Their result for densest subgraph builds on the multiplicative-weights-update $O(\log n)$-round approach of Bahmani, Goel, and Munagala~\cite{bahmani2014efficient}, while their result for $k$-core decomposition was inspired by the techniques of Esfandiari, Lattanzi, and Mirrokni~\cite{esfandiari2018parallel}.
The work of Mitrovi\'c, Pan, Qaempanah, and Raeisi~\cite{mitrovic2026new} focuses on computing a $(2+\eps)$-approximate \emph{directed} densest subgraph in $\tO(\sqrt{\log n})$ rounds in the sub-linear memory regime.
As mentioned in \cref{sec:application-densest-subgraph}, our densest-subgraph application builds on their fixed-threshold peeling algorithm.

In the introduction, we gave an overview of MPC algorithms whose round complexity substantially improves over what follows from direct PRAM simulation.
An influential line of work was initiated by Andoni, Song, Stein, Wang, and Zhong~\cite{andoni2018parallel} on computing connected components in MPC with bounds that depend on the graph diameter.
They and follow-up works~\cite{behnezhad2019near,liu2020connected,coy2022deterministic,fischer2022improved,balliu2023optimal} obtain algorithms running in $O(\log D + \log_{m/n} n)$ rounds in the sub-linear MPC regime, where $D$ is the diameter of the input graph.
This line of work is different in nature from our aim: its round complexity is parameterized by graph diameter and average degree, and in the worst case the $\log D$ term can still be $\Theta(\log n)$.
In contrast, our focus is on obtaining worst-case round-complexity improvements over the standard $O(\log n)$ bounds.

\subsection{Open problems}
Attaining a constant approximation in $O(\poly(\log\log n))$ rounds in the sub-linear MPC regime for any of the problems we studied remains a challenging open problem. 
Broadly speaking, we make progress for a family of peeling algorithms which use fixed peeling threshold.
It would be a major result to make comparable progress for peeling algorithms whose threshold is -- in some controlled way -- adjusted from one to another peeling step. Such approaches are used to solve $O(1)$-approximate maximum matching and $O(\log n)$-approximate minimum set cover.

\section{Preliminaries}

In this section, we provide a review of known results. 
Let $G = (V, E)$ be an undirected graph of $n$ vertices and $m$ edges.

\subsection{Basic notation} \label{sec:basic-definition}
Let $G = (V, E)$ be an undirected graph.
For a vertex $v \in V$, denote by $\deg_G(v)$ the degree of $v$ in $G$.
Let $N_G(v)$ denote the set of neighbors of $v$.

\begin{definition}[Edge orientation]
An \emph{orientation} of $G$ is a directed graph obtained by assigning each edge in $G$ a direction.
The \emph{out-degree} of an orientation $H$ is the maximum out-degree of a vertex in $H$.
We use $\alpha(G)$ to denote the minimum out-degree over all orientations of $G$.
The \emph{edge orientation problem} asks to find an edge orientation of out-degree $\alpha(G)$.
\end{definition}

\begin{definition}[Densest subgraph]
For a vertex set $S \subseteq V(G)$, we define the \emph{density} of $S$ to be $\rho(S) = |E(S)|/|S|$ where $E(S)$ is the set of edges in the induced subgraph of $S$.
Then, a \emph{densest subgraph} is a set $S^*$ such that $S^* \in \arg\max_{S\subseteq V(G)} \rho(S)$.
We let $\rho^*(G)$ be the density of the densest subgraph of $G$.
\end{definition}

\begin{definition}[Arboricity]
    The arboricity of $G$ is the minimum number of forests into which $E(G)$ can be partitioned.
    The arboricity is denoted by $\lambda(G)$, or simply $\lambda$ when the context is clear.
    It is known that $\lambda(G)$, $\alpha(G)$, and the density of the densest subgraph differ by at most 2.
\end{definition}

\begin{definition}[$k$-core decomposition]
Let $k \geq 1$ be an integer.
The \emph{$k$-core} of $G$ is the (unique) maximal subgraph $H$ of $G$ such that each vertex $v \in H$ has $\deg_H(v) \geq k$.
The \emph{coreness} of a vertex $v$, denoted by $C(v)$, is the largest integer $k$ such that $v$ is in the $k$-core.
It is known that the $k$-core is exactly the subgraph of $G$ induced by all vertices with coreness at least $k$.
An \emph{$\alpha$-approximate $k$-core decomposition} is a labeling $\tilde{C}: V(G) \rightarrow \{0, 1, \dots, n\}$ such that $\tilde{C}(v) \leq C(v) \leq \alpha \cdot \tilde{C}(v)$ holds for all $v \in V(G)$.
\end{definition}

For two functions $f(n), g(n)$, we write $f(n) = \tO(g(n))$ if $f(n) = O(g(n) \cdot \poly\log(g(n)))$. 
We denote $\log_2 n$ by $\lg n$ and define $\exp_2(x) = 2^x$.

\subsection{Partial layer assignments}
\begin{definition}[Layer Assignment and Vertex Set]
Let $G$ be a graph, and let $L$, $d$ be positive integers. A (partial) layer assignment of $G$ with $L$ layers and out-degree $d$ is a function $\ell_G: V(G) \rightarrow [L] \cup \{\infty\}$ such that for every vertex $v$ with $\ell_G(v)\neq \infty$, we have that
\[|u\in N_G(v) : \ell_G(u) \geq \ell_G(v)| < d.\]
The layer assignment is said to be \emph{full} if no vertex $v$ has $\ell_G(v) = \infty$.

We call the sets $V_i = \{v\in V(G) : \ell_G(v) \geq i\}$, for all $i \geq 1$, the vertex sets of the partial layer assignment $\ell_G$.
Denote by $V_{<i}$ the set $V(G) \setminus V_i$.
\end{definition}

\begin{definition}[Strictly Increasing Paths and Path Counts]
Let $G$ be a graph and $\ell_G$ be any partial layer assignment of $G$.
A path $P = (v_1, v_2, \ldots,v_{k})$ of length $k$ in $G$ is called strictly increasing (with respect to $\ell_G(v)$) if
\[\ell_G(v_1) < \ell_G(v_2) < \cdots < \ell_G(v_{k})\]
For each vertex $v \in V(G)$, define $\NumPathsIn_{G,\ell_G,k} (v)$ to be the number of distinct strictly increasing paths of length at most $k$ in $G$ that \textbf{end} at $v$, and $\NumPathsOut_{G,\ell_G, k} (v)$ to be the number of distinct strictly increasing paths of length at most $k$ in $G$ that \textbf{start} at $v$.
\end{definition}

In what follows, we leverage $\NumPathsIn_{G,\ell_G,k} (v)$ in bounding the size of vertex $v$'s neighborhood.
Compared to \cite{ghaffari2025density}'s definition, our $\NumPathsIn_{G,\ell_G,k} (v)$ counts only the paths of length at most $k$, instead of paths of an arbitrary length.
This allows us to tighten the analysis of our algorithms.


\begin{claim}
\label{claim:strictly-increasing-paths}
Let $G$ be a graph, and $d$, $k$ be positive integers. Let $\ell_G$ be any partial layer assignment with out-degree $d \geq 2$.
For any vertex $v \in V(G)$, the number of strictly increasing paths of length $k$ starting from $v$ is at most $d^{k-1}$.
\end{claim}
\begin{proof}
We prove this claim by an induction on $k$.
\begin{itemize}
    \item \textbf{Basis: $k = 1$.} Let $v$ be a vertex. There is only one path of length 1 starting from $v$, that is, the path consisting of $v$ itself. Hence, the claim holds for $k = 1$.
    
    \item \textbf{Inductive case:} Assume that the claim holds for $k = k' - 1$, where $k' > 1$.
    Fix a vertex $v$.
    A strictly increasing path from $v$ must first extend to a neighbor $v' \in N_G(v)$ with $\ell_G(v') > \ell_G(v)$; in addition, the rest of the vertices must form a strictly increasing path from $v'$.
    Since $\ell_G$ has out-degree $d$, there are only $d$ possible ways to choose $v'$.
    By the induction hypothesis, for a fixed $v'$, the number of strictly increasing paths from $v'$ of length $k'-1$ is at most $d^{k'-2}$.
    Hence, the number of strictly increasing paths of length $k'$ starting from $v$ is $d \cdot d^{k'-2} = d^{k' - 1}$.
    This proves the claim for $k = k'$.
\end{itemize}
\noindent By induction, the claim holds.
\end{proof}

\begin{lemma}
\label{lemma:mpc-numpaths}
Let $G$ be a graph, and $d$, $k$ be positive integers. Let $\ell_G$ be any partial layer assignment with out-degree $d \geq 2$.
Then, it holds that
\[\sum_{v\in V(G)}\NumPathsIn_{G, \ell_G, k}(v) \leq |V(G)|\cdot d^k.\]
\end{lemma}
\begin{proof}
Note that
\[
\sum_{v\in V(G)}\NumPathsIn_{G, \ell_G, k}(v) = \sum_{v\in V(G)}\NumPathsOut_{G, \ell_G, k}(v).
\]
For each vertex $v$, \cref{claim:strictly-increasing-paths} implies that $\NumPathsOut_{G, \ell_G, k}(v) \leq \sum_{i=1}^{k} d^{i-1} \leq d^{k}$.
Hence, we have $\sum_{v\in V(G)}\NumPathsIn_{G, \ell_G, k}(v) \leq \sum_{v \in V(G)} d^k \leq |V(G)| \cdot d^k$.
\end{proof}

\subsection{Pruned tree views}
On a high level, given a graph $G$, our algorithm repeatedly performs two operations: building a tree-like view within $G$ from a vertex, and pruning branches of a tree-like view.
As an illustration, a tree-like view from a vertex $v$ is performed by (1) creating a rooted tree consisting of $v$ only, and then (2) to each vertex $w$ in this tree attaching as its children all the neighbors of $w$. Step (2) is performed multiple times.
Observe that in this tree a vertex from $G$ might appear several times.
We now formally define when a tree-like view is valid, and we also define the notation of $\Missing$ neighbors, i.e., pruned branches.
\begin{definition}[Trees with Valid Mappings]
Let $G$ be a graph, $T$ be a rooted tree, and $\map : V(T)\rightarrow V(G)$ be a mapping from $T$ to $G$. 
We say that $\map$ is a valid mapping if
\begin{enumerate}
    \item for each edge $(x, y)$ in tree $T$, $\{\map(x),\map(y)\}$ is an edge in $G$, and
    \item for each vertex $x\in V(T)$, and any two distinct children $c_1$ and $c_2$ of $x$, it holds that $\map(c_1)\neq \map(c_2)$.
\end{enumerate}
\end{definition}
\begin{definition}[Missing Neighbors]
Let $G$ be a graph, $T$ be a rooted tree, and $\map : V(T) \rightarrow V(G)$ be a valid mapping. For a vertex $x\in V(T)$, we define
\[\Missing_{G, T, \map}(x) := |N_G(\map(x)) \backslash \{\map(c) : c\in \text{children}_T(x)\}|.\]
\end{definition}
\begin{definition}[Depth, Height]
    Let $T$ be a tree rooted at $r$. The \emph{depth} of a node $x \in V(T)$ is the number of edges in the unique path from $r$ to $x$. The \emph{depth} of $T$ is the largest depth over all vertices. The \emph{height} of $x$ is the number of vertices on the unique path from $x$ to the deepest node in its subtree. The \emph{height} of $T$ is defined as the height of its root.
\end{definition}

\begin{definition}[Monotonically Reachable]
    Let $G$ be a graph and let $\ell_G$ be a partial layer assignment. Further, let $T$ be a rooted tree with root $r$, and let $\map: V(T) \rightarrow V(G)$ be a valid mapping. A vertex $x \in V(T)$ is said to be \emph{monotonically reachable} with respect to $\ell_G$ if, when we denote by $x = x_1, x_2, \dots , x_k = r$ the unique path from $x$ to the root $r$ in $T$, it holds that
    \[
        \ell_G(\map(x_1)) < \ell_G(\map(x_2)) < \dots < \ell_G(\map(x_k)).
    \]
\end{definition}

We present the local pruning routine from \cite{ghaffari2025density} as \cref{alg:mpc:localprune}.
The idea of this routine is to use it while applying graph exponentiation to collect some $h$-hop neighborhood.
Normally, an $h$-hop neighborhood would be too large to store on a single machine in sub-linear MPC. 
However, performing graph exponentiation via \LocalPrune guarantees that -- for a large amount of vertices -- the desired $T$-hop neighborhoods fit within machine memory.
Specifically, the following lemmas hold.

\begin{algorithm}[h]
\caption{\LocalPrune: Pruning routine for rooted tree}
\label{alg:mpc:localprune}
\begin{algorithmic}[1]
\medskip
\Statex \textbf{Input:} rooted tree $T$, pruning parameter $k > 0$
\Statex \textbf{Output:} rooted tree $T_{\text{pruned}}$
\medskip
\hrule
\smallskip
\State $r \gets$ the root of $T$
\If{$r$ has at most $k$ children}
    \State\Return single-node tree with root $r$
\EndIf
\For{each child $c$ of $r$}
    \State $T_c\gets$ the subtree of $T$ rooted at $c$
    \State $T_{c,\text{pruned}} \gets \LocalPrune(T_c, k)$
\EndFor
\State $\mathcal{C} \gets \{T_{c,\text{pruned}} : c\text{ is a child of }r\}$
\State Remove the $k$ largest subtrees from $\mathcal{C}$
\State $T_{\text{pruned}}\gets$ single-node tree with root $r$
\For{each $T_{c,\text{pruned}} \in \mathcal{C}$}
    \State Attach $T_{c,\text{pruned}}$ as a subtree of $r$ in $T_{\text{pruned}}$
\EndFor
\State \Return $T_{\text{pruned}}$
\end{algorithmic}
\end{algorithm}

\begin{lemma}[Modified Lemma 3.2 from \cite{ghaffari2025density}]
\label{lemma:mpc-prunesize}
Let $G$ be a graph, and let $T$ be a rooted tree with root $r$ and height $L$. Suppose we have a valid mapping $\map:V(T) \rightarrow V(G)$. Let $d$ be a positive integer, and let $\ell_G$ be a partial layer assignment of $G$ with out-degree $d$ satisfying that $\ell_G(\map(r)) < \infty$. 
Furthermore, let $k\geq d$ be a positive integer, and obtain $T_\pruned$ by calling $\LocalPrune(T, \map, k)$. Then,
\[|V(T_\pruned)| \leq \NumPathsIn_{G,\ell_G, L}(\map(r)).\]
\end{lemma}
\begin{proof}
    We prove this lemma by induction on the height $L$ of $T$.
    \begin{itemize}
        \item \textbf{Basis: $L = 1$.}
        In this case, $T$ consists of a single node with no children.
        Thus, $\LocalPrune(T, k)$ returns $T$ itself in line 3.
        That is, $|T'| = |T| = 1$.
        On the other hand, since there is only one possible strictly increasing path with length one ending at $\map(r)$, we have
        \[
            \NumPathsIn_{G,\ell_G,1}(\map(r)) = 1 = |T'|.
        \]
        Hence, the lemma holds for $h_T = 1$.
        
        \item \textbf{Inductive case:}
        Assume that the lemma holds for all trees of height at most $L-1$, where $L \geq 2$.
        Let $c_1, c_2, \dots, c_p$ be the children of $r$.
        Denote by $T_{c_i}$ the subtree rooted at $c_i$.
        Since $\ell_G$ has out-degree $\leq k$, at most $k$ children $c_i$ have $\ell_G(\map(c_i)) \geq \ell_G(\map(r))$.
        Assume without loss of generality that $c_1, c_2, \dots, c_{k'}$ are these children, where $0 \leq k' \leq k$.
        
        In the for-loop of line 4, $\LocalPrune(T,k)$ computes $T_{c_i,\pruned} \gets \LocalPrune(T_{c_i},k)$ for $i \in [1,p]$.
        By the induction hypothesis,
        \[
            |T_{c_i,\pruned}| \leq \NumPathsIn_{G, \ell_G, L-1}(\map(c_i))
        \]
        holds for all $i \in [1, p]$.
        The largest $k$ trees among them are removed, and the remaining ones are attached to $r$ as a subtree.
        Therefore, we have
        \[
            |V(T_\pruned)| \leq 1 + \sum_{i=k'+1}^p \NumPathsIn_{G, \ell_G, L-1}(\map(c_i)).
        \]
        We claim that the term
        \[
            1 + \sum_{i=k'+1}^p \NumPathsIn_{G, \ell_G, L-1}(\map(c_i)) = \NumPathsIn_{G, \ell_G, L}(\map(r)).
        \]
        To see this, note that every strictly increasing path of length $L' \in [2, L]$ ending at $\map(r)$ can be written as $P \circ (c_i, v)$, where $P$ is a strictly increasing path of length $L' - 1$ ending at $c_i$ and $\circ$ denotes path concatenation. 
        Hence, $\sum_{i=k'+1}^p \NumPathsIn_{G, \ell_G, L-1}(\map(c_i))$ equals the number of strictly increasing paths ending at $\map(r)$.
        Since there is only one path of length 1 ending at $\map(r)$, the claim holds.
        This implies that
        \[
            |V(T_\pruned)| \leq \NumPathsIn_{G, \ell_G, L}(\map(r)).
        \]
        Consequently, the lemma also holds for trees of height $L$.
    \end{itemize}
    By induction, the lemma holds for trees of any height.
\end{proof}

\begin{lemma}[Claim 3.1 from \cite{ghaffari2025density}]
\label{lemma:mpc-prunemissing}
Let $G$ be a graph, $T$ be a rooted tree, and let $\map: V(T) \rightarrow V(G)$ be a valid mapping.
Fix any integer $k > 0$ and consider $T_\text{pruned} \gets \LocalPrune(T, k)$ with the induced mapping
\[
    \map_\text{pruned}(x) \defeq \map(x) \text{ for every } x \in V(T_\text{pruned}).
\]

Then, for every $x \in V(T_\text{pruned})$, it holds that
\[
    \Missing_{G,T_\text{pruned},\map_\text{pruned}} (x) \leq \Missing_{G,T,\map}(x) + k.
\]

\end{lemma}


\subsection{Concentration inequality}
We need the following standard results in probability.
\begin{lemma}[Chernoff Bound]\label{lemma:chernoff}
	Let $X_1, \ldots, X_k$ be independent random variables taking values in $[0, 1]$. Let $X \eqdef \sum_{i = 1}^k X_i$ and $\mu \eqdef \E{X}$. Then,
	\begin{enumerate}
		\item\label{item:delta-at-most-1} For any $\delta \in [0, 1]$ it holds $\prob{|X - \mu| \ge \delta \mu} \le 2 \exp\rb{- \delta^2 \mu / 3}$.
		\item\label{item:delta-at-most-1-ge} For any $\delta \in [0, 1]$ it holds $\prob{X \ge (1 + \delta) \mu} \le \exp\rb{- \delta^2 \mu / 3}$.
	\end{enumerate}
\end{lemma}

\section{The exponentiate-and-prune meta algorithm}
\subsection{Algorithm overview}
In this section, we present our exponentiate-and-prune meta method; see \cref{alg:mpc:expo} for details.
Given a graph $G$ and parameters $B, k, s, t$, the algorithm computes a rooted tree $T_v$ for each vertex $v$, called the \emph{pruned tree view} of $v$.
In addition, the algorithm marks some vertices as \emph{inactive}.
Intuitively, the pruned tree view of $v$ represents a neighborhood of $v$ that can be used to simulate peeling algorithms, and inactive vertices represent a small number of vertices whose pruned tree view is too large to be computed.
Hence, to implement this method, two components are required: (1) simulating peeling on the pruned tree views, and (2) handling inactive vertices.
These components are application dependent, and some of their instantiations for concrete problems are presented in detail in \cref{sec:orientation,sec:coloring,sec:densest-subgraph,sec:k-core}.

\begin{algorithm}
\caption{\Expo}
\label{alg:mpc:expo}
\begin{algorithmic}[1]
\medskip
\Statex \textbf{Input:} graph $G$, pruning parameter $k > 0$, step parameters $s > 0$ and $t \geq 0$, and budget $B > 0$
\Statex \textbf{Output:} for each $v\in V(G)$, a rooted tree $T_v$ with a valid mapping $\map_v : V(T_v)\rightarrow V(G)$
\medskip
\Statex \hrule
\State \textbf{Initialization:}
\For{each $v\in V(G)$ with $|N_G(v)| < B^{1/2^s} + k$} \Comment{initialize active vertices}
    \State Mark $v$ as active
    \State $T_v^{(0,0)} \gets$ a rooted tree such that the root has $|N_G(v)|$ leaf children \label{line:init-active-tree}
    \State Construct $\map_v^{(0,0)} : V(T_v^{(0,0)}) \rightarrow V(G)$ so that root of $T_v^{(0,0)}$ maps to $v$, and each child maps to a distinct neighbor in $N_G(v)$ \label{line:init-active-map}
\EndFor
\For{each $v\in V(G)$ with $|N_G(v)| \geq B^{1/2^s} + k$} \label{line:init-inactive} \Comment{initialize inactive vertices}
    \State Mark $v$ as inactive \label{line:mark-inactive-1}
    \State $T_v^{(0,0)} \gets$ rooted tree with a single node
    \State Define $\map_v^{(0,0)} : V(T_v^{(0,0)}) \rightarrow V(G)$ so that the single node maps to $v$
\EndFor

\Statex

\For{$i = 0, 1, \dots, t$} \Comment{each iteration exponentiate $s$ times and performs a pruning step} \label{line:i-loop}
    \State \textbf{Exponentiation step:}
    \For{$j = 1, 2, \dots, s$} \label{line:j-loop} \Comment{$s$ exponentiation steps}
        \State $T_v^{(i,j)} \gets T_v^{(i,j-1)}$ for each $v \in V(G)$
        \State $\map_v^{(i,j)} \gets \map_v^{(i,j-1)}$ for each $v \in V(G)$
        \If{$v$ is active}
            \State Let $x_1, x_2, \dots, x_{\eta}$ be the leaves of $T_v^{(i,j)}$ that are mapped to active vertices and have depth exactly $2^{i\cdot s+j-1}$
            \For {each $x_b$} \label{line:exponentiate}
                \State $u_b \gets \map_{v}^{(i,j-1)}(x_b)$
                \State Attach $T_{u_b}^{(i,j-1)}$ at the leaf $x_b$ of $T_{v}^{(i,j)}$ \label{line:attach-tree}
                \State Expand $\map_{v}^{(i,j)}$ so that it maps each newly attached node $x$ to $\map_{u_b}^{(i,j-1)}(x)$
            \EndFor
        \EndIf
    \EndFor
    \State \textbf{Pruning step:}
    \If {$i \neq t$} \Comment{perform the pruning except in the last iteration}
        \For{each $v\in V(G)$} 
            \State $T_{v}^{(i+1,0)}\gets \LocalPrune\rb{T_v^{(i,s)}, k}$ \label{line:invoke-prune}
            \State Define $\map_{v}^{(i+1,0)} : V\rb{T_{v}^{(i+1,0)}} \rightarrow V(G)$ to be $\map_{v}^{(i,s)}$ restricted to $T_{v}^{(i+1,0)}$
            \If{$|V\rb{T_{v}^{(i+1,0)}}| \geq B^{1/2^s}$}
                \State mark $v$ as inactive \label{line:mark-inactive-2}
            \EndIf
        \EndFor
    \EndIf
    
\EndFor
\State \Return $T_v^{(t,s)}$ and $\map_v^{(t,s)}$ for every $v\in V(G)$
\end{algorithmic}
\end{algorithm}

The input parameters $B, k, s, t$ to \Expo are application dependent, and are detailed as follows.
\begin{itemize}
    \item $B$ is the \textit{budget parameter} that controls the maximum size of pruned trees. 
    It affects memory consumption and the number of inactive vertices.
    \item $s$ and $t$ are the \emph{step parameters}. On a high level, \Expo performs $s$ graph exponentiation steps, and after calls \LocalPrune to reduce the tree size.
    This process is repeated $t + 1$ times with an exception that the last repetition does not call \LocalPrune.
    Hence, the total number of exponentiation steps is $s(t+1)$ and the number of \LocalPrune calls is $t$.
    \item $k > 0$ is the \emph{pruning parameter} passed to \LocalPrune. For technical reasons, we require that $B > k^{2^s}$. See \cref{rmk:local-space} for details.
\end{itemize}

\noindent 
The properties of the algorithm are summarized in the following theorem.

\begin{theorem}
\label{theorem:missing}
Let $G$ be a graph, and let $k > 0$, $s > 0$, $t \geq 0$, and $B > k^{2^s}$ be integers.
Let $T_v$ and $\map_v$ be the rooted tree and mapping obtained by invoking $\Expo$ on $G$ with parameters $k, s, t$, and $B$.
All vertices $v \in V(G)$ satisfy the following.
\begin{itemize}
    \item[] (P1) For every node $x\in V(T_v)$ with depth $< 2^{s(t+1)}$ that maps to an active vertex, it holds that
        \[|\Missing_{G, T_v,\map_v}(x)| \leq t\cdot k.\]
    \item[] (P2) Let $\ell_G$ be any partial layer assignment of out-degree $k$. If $v$ is inactive and $\ell_G(v) \neq \infty$, then $\LazyPathsIn(v) \geq B^{1/2^{s}}$, where $L = 2^{s(t+1)}$.
\end{itemize}
The algorithm runs in $O(s \cdot (t + 1))$ rounds in the sublinear MPC model using $O(B \cdot 2^{2^s}) = B^{1+o(1)}$ local space and $O(n\cdot B \cdot 2^{2^s} + m)$ global space.
\end{theorem}

\begin{remark} \label{rmk:local-space}
In \cref{theorem:missing}, we assume that $B > k^{2^s}$. We remark that the $k^{2^s}$ term is at most $n^{\delta'}$ for some small constant $\delta' \in (0, 1)$ in all our settings. The $2^{2^s}$ term in the local space is $n^{o(1)}$ in all our settings.
\end{remark}

We now discuss the implications and meaning of the two properties in \cref{theorem:missing}.
It is instructive to think of $T_v$ as a pruned $2^{s(t+1)}$-hop neighborhood of $v$.
In this context, ``pruning'' means that some neighbors of a vertex are ignored.
Property~P1 is saying that at most $t \cdot k$ neighbors are ignored.
Later, we convert this into an approximation guarantee.
In the process of constructing $T_v$ for all vertices $v$, some of those trees might become ``too large'' to be efficiently handled by the algorithm: a vertex $v$ whose $T_v$ is too large is called inactive.
Property~P2 lower-bounds size of trees of inactive vertices, which we use to upper-bound the number of inactive vertices.

Implementation details are as follows.
The pruned trees $\{T_v\}_{v \in V(G)}$ are computed by combining the standard graph exponentiation with the pruning procedure (\LocalPrune).
We use the notation $T_v^{(i, j)}$ to denote the tree of $v$ obtained after $i \cdot s + j$ exponentiation steps and $i$ pruning steps.
Intuitively, $T_v^{(i, j)}$ represents the $2^{i\cdot s+j}$-hop pruned neighborhood of $v$.

Lines 1-9 initialize these trees.
For vertices with degree $< B^{1/2^s} + k$, the for-loop in line 2 initializes $T_v^{(0, 0)}$ to represent the 1-hop neighborhood of $v$.
All vertices with degree $\geq B^{1/2^s} + k$ are marked as inactive in line 7.
Once a vertex $v$ is marked as inactive, the algorithm stops expanding its tree.

Consider the for-loop in line 10.
Each iteration $i$ performs $s$ graph exponentiation steps on the trees of active vertices (see lines 11-20).
In addition, if $i \neq t$, a pruning step is performed on each tree (lines 21-27).
Fix an iteration $i$.
The for-loop in line 12 implements the graph exponentiation procedure.
More precisely, the goal of an iteration $j$ is to compute the $2^{i\cdot s + j}$-hop pruned neighborhoods (trees) $T_v^{(i, j)}$ from $2^{i\cdot s + j - 1}$-hop pruned neighborhoods.
See lines 12-20 for details of this computation.
The pruning step (lines 22-27) is done by invoking \LocalPrune on each tree.
If the tree of a vertex $v$ has more than $B^{1/2^s}$ vertices after pruning, line 27 marks it as inactive.

\subsection{Proof of \cref{theorem:missing}}

Consider an execution of \Expo$(G, k, s, t, B)$ for a fixed set of parameters $k, s, t,$ and $B$.
The proof of \cref{theorem:missing} consists of three parts.
\begin{enumerate}
    \item \textbf{The proof of (P1):} This part analyzes the number of missing neighbors of a tree node that maps to an active vertex. See \cref{sec:missing-bound} for details.
    \item \textbf{The proof of (P2):} This part lower bounds $\LazyPathsIn(v)$ for each inactive vertex $v$ with respect to a partial layer assignment $\ell_G$. See \cref{sec:pathin-bound}.
    \item \textbf{Space and round complexities:} This part is presented in \cref{sec:complexity-bound}.
\end{enumerate}

\noindent The proof is summarized in \cref{sec:missing-theorem-summary}.

\subsubsection{Missing bounds} \label{sec:missing-bound}

\begin{definition}[Missing Bound] \label{def:missing-bound}
Let $T$ be a rooted tree associated with a valid mapping $\map$.
We say $T$ has a \emph{missing bound} of $c$ for $d$ layers, if
\[
    \Missing_{G,T,\map}(x) \leq c
\]
for every $x \in V(T)$ which maps to an active vertex and whose depth is smaller than $d$.
\end{definition}

The following lemma shows a missing bound for each $T_v^{(i,j)}$.

\begin{lemma} \label{lemma:missing}
    The following holds for all $v \in V(G)$. For all $i \in [0, t]$ and $j \in [0, s]$, $T_v^{(i,j)}$ has a missing bound of $i \cdot k$ for $2^{i\cdot s + j}$ layers.
\end{lemma}

To prove the lemma, we need the following claim showing that each exponentiation step doubles the layer guarantee in the missing bound.

\begin{claim} \label{clm:missing-exponentiate}
    Consider a fixed $i \in [0, t]$ and let $c \geq 0$ be a non-negative integer.
    Suppose that, for all vertices $v \in V(G)$, $T_v^{(i,0)}$ has a missing bound of $c$ for $2^{i\cdot s}$ layers.
    Then, for all vertices $v \in V(G)$ and $j \in [0,s]$, $T_v^{(i, j)}$ has a missing bound of $c$ for $2^{i\cdot s + j}$ layers.
\end{claim}
\begin{proof}
    Consider a fixed $i$. We prove this claim by induction on $j$.
    \begin{itemize}
        \item \textbf{Basis: $j = 0$.}
        The basis trivially holds because the claim assumes that $T_v^{(i,0)}$ has a missing bound of $c$ for $2^{i \cdot s} = 2^{i \cdot s + j}$ layers.
        \item \textbf{Inductive case:}
        Suppose that the claim holds for $j = j' - 1$, where $j' \geq 1$.
        To ease the notation, we denote $T_v^{(i,j')}$ by $T$ and $\map_v^{(i,j')}$ by $\map$.
        \cref{alg:mpc:expo} computes $T$ as follows.
        (See the for-loop in line~\ref*{line:exponentiate}.)
        First, $T$ is initialized as a copy of $T^{(i,j'-1)}$. Next, we obtain the leaves $x_1, x_2, \dots, x_\eta$ of $T$ that are mapped to active vertices and have distance exactly $2^{i\cdot s + j' - 1}$.
        Then, we compute $u_b = \map_v^{(i,j'-1)}(x_b)$ and attach the tree $T_{u_b}^{(i,j'-1)}$ on $x_b$.
        The mapping $\map$ is obtained by combining $\map_v^{(i,j'-1)}$ and $\{\map_{u_b}^{(i,j'-1)}\}_{b\in [\eta]}$.

        Since $T$ is initially a copy of $T^{(i,j'-1)}$, each node $x \in V(T)$ with depth $< 2^{i \cdot s + j'-1}$ corresponds to a node in $T^{(i,j'-1)}$.
        As the child sets of $x$ and $x'$ are the same, they have the same number of missing neighbors.
        Since the claim holds for $j'-1$, $T_v^{(i,j'-1)}$ has a missing bound of $c$ for $2^{i \cdot s + j' - 1}$ layers.
        Therefore, we have
        \begin{itemize}
            \item (R1) $\Missing_{G,T,\map}(x) \leq c$ for every $x \in V(T)$ whose depth is smaller than $2^{i\cdot s + j'-1}$ that maps to an active vertex.
        \end{itemize}

        Consider now a node $x \in V(T)$ whose depth is within $[2^{i\cdot s + j'-1}, 2^{i\cdot s + j'})$.
        By construction, $x$ is contained in the subtree of some node $x_b$, $b \in [\eta]$.
        Note that this subtree is a copy of $T_{u_b}^{(i,j'-1)}$.
        Since $T_{u_b}^{(i,j'-1)}$ has a missing bound of $c$ for $2^{i \cdot s + j' - 1}$ layers, we have $\Missing_{G,T,\map}(x) \leq c$ if $x$ maps to an active vertex.
        Applying this argument on all $x$ whose depth from the root is within $[2^{i\cdot s + j'-1}, 2^{i\cdot s + j'})$, we know that

        \begin{itemize}
            \item (R2) $\Missing_{G,T,\map}(x) \leq c$ for every $x \in V(T)$ whose depth is in $[2^{i\cdot s + j'-1}, 2^{i\cdot s + j'})$ that maps to an active vertex.
        \end{itemize}
        
        By (R1) and (R2), $T = T_v^{(i,j)}$ has a missing bound of $c$ for $2^{i\cdot s + j'}$ layers.
        By induction, the claim holds for all $j$.
    \end{itemize}
\end{proof}

\noindent We now prove \cref{lemma:missing} using \cref{clm:missing-exponentiate}.

\begin{proof}[Proof of \cref{lemma:missing}]
    We prove the lemma by induction on $i$.
    \begin{itemize}
        \item \textbf{Basis: $i = 0$.} Consider a vertex $v$. We need to show that $T_v^{(0,0)}$ has a missing bound of $i \cdot k = 0$ for $2^{i \cdot s} = 1$ layers. Since $T_v^{(0,0)}$ is a rooted tree of height 2, the only vertex with depth $< 1$ is the root itself. If the root is mapped to an active vertex, then lines~\ref*{line:init-active-tree} and \ref*{line:init-active-map} of \cref{alg:mpc:expo} would add to $T^{(0,0)}_v$ all neighbors of $v$.
        Hence, $T_v^{(0,0)}$ has a missing bound of $0$ for $2^0$ layers.
        By applying \cref{clm:missing-exponentiate} with $i = 0$, we know that $T_v^{(0,j)}$ has a missing bound of $0$ for $2^j$ layers.
        This completes the proof for $i = 0$.
        
        \item \textbf{Inductive case:}
        Assume that the lemma holds for $i = i'-1$ for some $i' \geq 1$.
        Since $i' \leq t+1$, the tree $T_v^{(i',0)}$ is computed in line~\ref*{line:invoke-prune} as \LocalPrune$(T_v^{(i'-1,s)}, k)$.
        By the induction hypothesis, we know that $T_v^{(i' - 1, s)}$ has a missing bound of $(i'-1) \cdot k$ for $2^{(i'-1) \cdot s + s} = 2^{i' \cdot s}$ layers.
        By \cref{lemma:mpc-prunemissing}, the invocation of \LocalPrune only increases the missing set size of every node in $T_v^{(i'-1,s)}$ by $k$.
        Hence, $T_v^{(i',0)}$ has a missing bound of $i' \cdot k$ for $2^{i' \cdot s}$ layers.
        Applying \cref{clm:missing-exponentiate} with $i = i'$, we know that $T_v^{(i',j)}$ has a missing bound of $i' \cdot k$ for $2^{i' \cdot s + j}$ layers.
        Hence, the lemma also holds for $i = i'$.
    \end{itemize}

    By induction, the lemma holds.
\end{proof}

\subsubsection{Bounds on $\NumPathsIn$} \label{sec:pathin-bound}

\begin{lemma}
\label{lemma:expo-pathinbound}
    Let $L = 2^{s(t+1)}$.
    Consider any partial layer assignment $\ell_G$ of out-degree at most $k$. At the end of \Expo, every inactive vertex with $\ell_G(v) \neq \infty$ has $\LazyPathsIn(v) \geq B^{1/2^s}$.
\end{lemma}
\begin{proof}
    During the execution, \cref{alg:mpc:expo} only marks a vertex in the following two lines.
    \begin{itemize}
        \item During initialization, a vertex $v$ is marked as inactive in line~\ref*{line:mark-inactive-1} when it has at least $B^{1/2^s} + k$ neighbors. 
        Since $\ell_G$ has out-degree $\leq k$, we know that at least $B^{1/2^s}$ of these neighbors $u$ satisfy $\ell_G(u) < \ell_G(v)$. This directly implies that $\LazyPathsIn(v) \geq B^{1/2s}$.
        
        \item Consider an iteration $i$ of the for-loop in line~\ref*{line:i-loop}. After invoking \LocalPrune, line~\ref*{line:mark-inactive-2} marks a vertex $v$ as inactive if the size of its tree, $|T_v^{(i+1,0)}|$, is still at least $B^{1/2^s}$.
        Since \LocalPrune is invoked with parameter $k$, the out-degree of $\ell_G$ is at most $k$, and $\ell_G(v) \neq \infty$, \cref{lemma:mpc-prunesize} ensures that the tree returned by \LocalPrune has size at most $\LazyPathsIn(v)$.
        Hence, only vertices $v$ with $\LazyPathsIn(v) \geq B^{1/2^s}$ are marked as inactive in this line.
    \end{itemize}
    \noindent Since in both cases a vertex $v$ is only marked as inactive if $\LazyPathsIn(v) \geq B^{1/2^s}$, the lemma holds.
\end{proof}

\subsubsection{Space and round complexities} \label{sec:complexity-bound}

\begin{lemma} \label{lemma:tree-size}
    Consider an iteration $(i, j)$ of the for-loop in line~\ref{line:j-loop}.
    For every active vertex $v$, we have $|T_v^{(i,j)}| \leq 2^{2^j} \cdot B^{2^j/2^s}$ for all $v \in V(G), i \in [0, t],$ and $j \in [0,s]$.
\end{lemma}
\begin{proof}
    Consider a fixed $i$.
    We prove this lemma by induction on $j$.
    \begin{itemize}
        \item \textbf{Basis: $j = 0$.}
        Consider a vertex $v$.
        If $i = 0$, $T^{(i,0)}_v$ is active only if $v$ has $< B^{1/2^s} + k \leq 2B^{1/2^s}$ neighbors, where the last inequality is due to $B > k^{2^s}$.
        If $i > 0$, $T^{(i,0)}_v$ is active only if its size is at most $B^{1/2^s}$. (See line~\ref*{line:mark-inactive-2}.)
        Hence, $T_v^{(i,0)} \leq 2B^{1/2^s} = 2^{2^j} \cdot B^{2^j/2^s}$ holds.
        
        \item \textbf{Inductive case:} 
        Suppose that the lemma holds for $j = j' - 1$ for some $j' \geq 1$.
        In iteration $(i, j')$, an active vertex $v$ computes its tree $T_v^{(i,j')}$ by attaching at most $|T_v^{(i,j'-1)}|$ trees to $T_v^{(i,j'-1)}$, where each attached tree is $T_u^{(i,j'-1)}$ for some $u \in V(G)$. Let $B' = 2^{2^{j'-1}} \cdot B^{2^{j'-1}/2^s}$. By the induction hypothesis, at most $|T_v^{(i,j'-1)}| \leq B'$ trees are attached, and each attached tree is of size $\leq B'$. Consequently, the resulting tree has size at most $(B')^2 \leq 2^{2^{j'}} \cdot B^{2^{j'}/2^s}$. Hence, the lemma also holds for $j = j'$.
    \end{itemize}
    By induction, the lemma holds.
\end{proof}

\cref{lemma:tree-size} implies that the tree size of an active vertex is always upper bounded by $2^{2^s} \cdot B$.
In our algorithm, a vertex stops growing its tree once it is marked inactive. Hence, we have the following.

\begin{corollary} \label{corollary:tree-size}
    For all $v \in V(G), i \in [0, t]$, and $j \in [0, s]$, $|T_v^{(i,j)}| \leq 2^{2^s} \cdot B$.
\end{corollary}

\begin{lemma} \label{lemma:implementation}
\Expo can be implemented in $O(s(t+1))$ MPC rounds with $O(2^{2^s} \cdot B)$ words of local memory and $O(2^{2^s} \cdot n \cdot B + m)$ words of global memory, where $n$ and $m$ are upper bounds on the number of vertices and edges, respectively.
\end{lemma}
\begin{proof}
\cref{corollary:tree-size} ensures that the tree of each vertex fits into the local memory, and therefore the exponentiation step (line~\ref*{line:attach-tree}) can be performed in $O(1)$ rounds of communication.
The rest of the algorithm can be implemented in a straightforward way using known MPC primitives developed in previous work.
See, for example, \cite[Section E]{andoni2018parallel} and the references therein.
\end{proof}

\subsubsection{Summary} \label{sec:missing-theorem-summary}

\begin{proof}[Proof of \cref{theorem:missing}]
    Consider an execution of \Expo$(G, k, s, t, B)$ for a given set of parameters $k, s, t$ and $B$.
    By \cref{lemma:missing}, for each vertex $v$, the returned tree $T_v^{(t,s)}$ has a missing bound (\cref{def:missing-bound}) of $t \cdot k$ for $2^{s(t+1)}$ layers, which proves (P1).
    Property (P2) is ensured by \cref{lemma:expo-pathinbound}.
    The round and memory complexities are shown in \cref{lemma:implementation}.
    Hence, the theorem holds.
\end{proof}

\section{Edge orientation} \label{sec:orientation}
We demonstrate how to use our new exponentiate-and-prune algorithm to compute a $(2+\eps)(t+1)$-approximate edge orientation for any integer parameter $t \geq 0$.

\subsection{Base algorithm}
The main idea of computing an approximate edge orientation is to simulate the peeling algorithm from \cite{barenboim2008sublogarithmic}.
We present it in \cref{alg:peeling-orientation}.

\begin{algorithm}
\caption{Finds $(2+\eps)$-approximate edge orientation}
\label{alg:peeling-orientation}
\textbf{Input:} graph $G$, parameters $\eps$, $k$
\begin{algorithmic}[1]
\For{$i = 1$ to $\lceil\log_{1+\eps} n\rceil$}

    \State $A\gets$ all vertices $v\in V(G)$ with $d_{G}(v) \leq (2+2\eps)k$
    \For{each $v\in A$}
        $\ell(v)\gets i$
    \EndFor
    \State Remove vertices in $A$ from $G$
\EndFor
\State \Return $\ell$
\end{algorithmic}
\end{algorithm}

The algorithm iteratively removes all vertices with degree less than $(2+2\eps)k$. If $k\geq \alpha(G)$, it is known that all vertices will be removed within $\log_{1+\eps} n$ iterations.
Additionally, it computes a full layer assignment $\ell$ with out-degree at most $(2+2\eps)k$.
By doing so, the layer assignment gives us an edge orientation with maximum out-degree $(2+2\eps)k$: we direct every edge from the vertex with a smaller layer to the vertex with a larger layer (if the layers are the same, it can be arbitrarily directed).
This bounds the out-degree of any vertex by $(2+2\eps)k$.
We present these properties in the following lemmas.

\begin{lemma}[Folklore, see, e.g., \cite{ghaffari2025density}]
\label{lemma:orientation:peeling}
Let $G$ be a graph, $\eps\in (0, 1)$, and $k\geq \alpha(G)$. Then, \cref{alg:peeling-orientation} removes all the vertices in $G$.
\end{lemma}
\begin{proof}
We know that $\alpha(G)$ is lower bounded by the density of the densest subgraph of $G$.
Then, the total number of vertices in $G$ that can have degree greater than $(2+2\eps)k$ is
\[\frac{2|E(G)|}{(2+2\eps)k} \leq \frac{|V(G)|}{1+\eps}.\]
Therefore, the vertex set size of the graph decreases by a factor of at least $(1+\eps)$ every iteration, resulting in $\lceil \log_{1+\eps} n\rceil$ iterations before removing all the vertices in the graph.
\end{proof}

\begin{lemma}[\cite{ghaffari2025density}]
Let $G$ be a graph, $\eps\in (0, 1)$, and $k\geq \alpha(G)$. Then $\ell$ returned by \cref{alg:peeling-orientation} is a partial layer assignment with out-degree $(2+2\eps)k$.
\end{lemma}
\begin{proof}
Consider a vertex that is removed during iteration $i$.
It must have degree at most $(2+2\eps)k$ to be removed during iteration $i$.
Additionally, its neighbors during iteration $i$ all have layers at least $i$.
Therefore, $\ell$ is a partial layer assignment with out-degree $(2+2\eps)k$.
\end{proof}
So, when $k$ is set to $\alpha(G)$, we are able to compute a $(2+\eps)$-approximate edge orientation using \cref{alg:peeling-orientation}.

\subsection{Edge partitioning}
To use our exponentiate and prune algorithm, we want $\alpha(G)$ to be $O(\log n)$ for our graph $G$.
We accomplish this by partitioning the edges of the graph uniformly at random.
We have the following lemma.
\begin{lemma}[\cite{ghaffari2025density}]
\label{lemma:edge-partition}
Let $G$ be a graph, $k \geq \lceil \alpha(G) \rceil$, and $\eps\in(0, 1)$. Let $H = \lceil \frac{\eps^2 k}{100\log n} \rceil$ and $G_1, G_2,\ldots, G_H$ be the graphs from partitioning the edges into $H$ groups uniformly at random. Then, whp $\max_{i\in[H]} \alpha(G_i) \leq \frac{(1+\eps)k}{H} = O(\log n/\eps^2)$.
\end{lemma}
\begin{proof}
If $H = 1$, $\alpha(G)$ is bounded by $\frac{100\log n}{\eps^2}$.
So, we assume that $H > 1$. 
Consider an orientation of the edges of $G$ with out-degree $k$.
For any vertex $v$, let us say it is in $G_i$ for some $i$.
Then, we have that
\[\E{\deg_{G_i}^{out}(v)} \leq \frac{k}{H}\]
which is the expected number of out-neighbors of $v$ that are also in $G_i$.
Using the Chernoff bound, we see that
\[\prob{\deg_{G_i}^{out}(v) \geq (1+\eps)k/H} \leq \exp\rb{-\eps^2k/(3H)} \leq n^{-6}.\]
Therefore, taking the union bound over all $v$, we have that each vertex will have at most $(1+\eps)k/H$ out-neighbors in the same group with probability at least $1 - n^{-5}$.
This gives us our desired result.
\end{proof}

\subsection{Setting parameters $L$ and $s$}
Our exponentiate and prune algorithm uses the parameters $B$, $k$, $s$, $t$ in order to produce a pruned tree view for every vertex of the graph.
Using these pruned tree views, we want to simulate as many iterations of peeling algorithms specific to each problem; we call the number of iterations $L$, representing the number of layers in the partial layer assignment that are peeled.
Therefore, we focus on the budget parameter ($B$), the out-degree/peeling parameter ($k$), and the step and approximation parameter ($t$).
Using these parameters, we can set the remaining variables $s$ and $L$ in terms of them.
Now, we define the following two properties we need for our applications:
\begin{enumerate}
    \item \textbf{(Upper bound on inactive vertices)} The total number of inactive vertices with $\ell_G(v) < \infty$ must be less than or equal to $n/2^L$ for any partial layer assignment $\ell_G$.

    \item \textbf{(Lower bound on exponentiation)} The number of exponentiation steps of method \Expo must be at least $\lg L$.
\end{enumerate}
The first property ensures that the inactive vertices do not disturb the applicable peeling algorithms much.
The second property guarantees that the pruned tree views have enough depth in order to properly simulate peeling the first $L$ layers of the partial layer assignment.
With these two properties in mind, we present the following lemma that specifies the values of $s$ and $L$.

\begin{lemma}
\label{lemma:parameters}
Let $G$ be a graph, $B,k,t\in\mathbb{N}$. We set
\[L = \left\lfloor\frac{\lg^{\frac{t+1}{t+2}}\rb{B}}{4\lg(2k)}\right\rfloor\text{ and } s = \left\lceil \frac{\lg L}{t + 1}\right\rceil.\]
Let $\ell_G$ be any partial layer assignment with out-degree $k$. Then, the following hold.
\begin{enumerate}
    \item The number of inactive vertices $v$ with $\ell_G(v) < \infty$ is at most $\frac{n}{2^L}$.
    \item $\Expo(G, k, s, t, B)$ has at least $\log L$ exponentiation steps.
\end{enumerate}
\end{lemma}
\begin{proof}
We first show that the number of inactive vertices $v$ with $\ell_G(v) < \infty$ is at most $\frac{n}{2^L}$.
Let $L' = 2^{s(t+1)}$.
Note that $L \leq L' \leq 2L$.
Recall that $V_{< \infty}$ is the set of vertices $v$ with $\ell_G(v) < \infty$
Using \cref{lemma:mpc-numpaths}, we have that 
\[
\sum_{v\in V_{<\infty}} \NumPathsIn_{G,\ell_G, L'}(v) \leq |V_{<\infty}|\cdot k^{L'}.
\]
Also, using \cref{theorem:missing}, we know that every inactive vertex $v$ with $\ell_G(v) < \infty$ has $\NumPathsIn_{G,\ell_G,L'}(v) \geq B^{1/2^s}$.
Hence, the number of inactive vertices in $V_{<\infty}$ is at most
\[
\frac{|V_{<\infty}|\cdot k^{L'}}{B^{1/2^s}} \leq \frac{nk^{L'}}{B^{1/2^s}}.
\]
To upper bound the right hand side by $\frac{n}{2^L}$, it suffices to verify that $B^{1/2^s} \geq (2k)^{L'}$, or equivalently,
\begin{equation}\label{eqn:verify}
    B \geq (2k)^{L' \cdot 2^s}.
\end{equation}
By our choice of $L$ and $s$, $2^s \leq 2\lg^{1/(t+2)}(B)$.
Recall that $\exp_2(x)$ denotes $2^x$.
Write
\begin{align*}
    (2k)^{L' \cdot 2^s} &= \exp_2\rb{\lg(2k) \cdot L' \cdot 2^s} \\
    & \leq \exp_2\rb{\lg(2k) \cdot \frac{1}{4\lg(2k)} \cdot 2\lg^{(t+1)/(t+2)}(B) \cdot 2\lg^{1/(t+2)}(B)} \\
    & = \exp_2(\lg B) = B.
\end{align*}
Consequently, \cref{eqn:verify} indeed holds.
Therefore, the number of inactive vertices with $\ell_G(v) < \infty$ is upper bounded by $\frac{n}{2^L}$.

For the number of exponentiation steps in $\Expo(G, k, s, t, B)$, it takes $s(t+1)$ exponentiation steps which is at least $\lg L$.
\end{proof}


\subsection{Local vertex layers}
Now, we perform peeling on the local pruned tree views and produce local vertex layers.
We present \PartialLayerTree from \cite{ghaffari2025density} in \cref{alg:mpc:partiallayertree}.

\begin{algorithm}
\caption{(\PartialLayerTree) Finds partial layer assignment for rooted tree}
\label{alg:mpc:partiallayertree}
\begin{algorithmic}[1]
\medskip
\Statex\textbf{Input:} graph $G$, rooted tree $T$, valid mapping $\map:V(T)\rightarrow V(G)$, and $a, L\in \mathbb{N}$

\Statex\textbf{Output:} partial layer assignment $\ell_T : V(T) \rightarrow [L] \cup \{\infty\}$
\medskip
\hrule
\medskip
\State $\Missing(x) \gets |N_G(\map(x))| - |\children_T(x)|$ for every $x\in V(T)$
\State $\ell_T(x)\gets \infty$ for all $x\in V(T)$
\For{$i = 1$ to $L$}
    \State $A\gets$ all vertices $x\in V(T)$ with $|\children_T(x)| + \Missing(x) < a$
    \For{each $x\in A$}
        $\ell_T(x)\gets i$
    \EndFor
    \State Remove vertices in $A$ from $V(T)$
\EndFor
\State \Return $\ell_T$
\end{algorithmic}
\end{algorithm}

In \cref{alg:mpc:partiallayertree}, we are using the parameter $a$ as our peeling threshold.
The main idea is to set $a = (t + 1)\cdot k$ after using our \Expo algorithm.
The intuition is that \cref{theorem:missing} tells us that $\Missing(x) \leq t\cdot k$ when $\map(x)$ is an active vertex.
Therefore, \PartialLayerTree will peel all active vertices with degree at most $k$, and it will never peel any vertex with degree at least $(t+1)\cdot k$.
This peeling property gives us the $(t+1)$-approximation factor in our results.
We now present the following lemma about \PartialLayerTree.

\begin{lemma}[Modified Lemma $3.9$ from \cite{ghaffari2025density}]
\label{lemma:orientation:partiallayer}
Let $G$ be a graph, $d, L\in\mathbb{N}$, and $\ell_G:V(G)\rightarrow [L]\cup \{\infty\}$ be a partial layer assignment with out-degree $d$.
Let $B, k, s, t$ be parameters that satisfy $k\geq d$ and $s(t+1) \geq \log L$.
Let $(T_v, \map_v)$ be the rooted tree and mapping attained from $\Expo(G, k, s, t, B)$ and $v$ is any $v\in V(G)$ with $\ell_G(v) \neq \infty$ and $\NumPathsIn_{G,\ell_G, L'}(v) < B^{1/2^s}$ where $L' = 2^{s(t+1)}$.
Also, let $\ell_v$ be the partial layer assignment returned by $\PartialLayerTree(G, T_v, \map_v, (t+1)\cdot k, L)$.
Then, for every $x\in V(T)$ that is monotonically reachable with respect to $\ell_G$, we have $\ell_{v}(x) \leq \ell_G(\map_v(x))$.
\end{lemma}
\begin{proof}
We refer the reader to the proof of Lemma $3.9$ in \cite{ghaffari2025density}, but we use \cref{theorem:missing} to say that $\Missing(x) \leq t\cdot k$ for any vertex $x\in V(T_v)$ with $\map_v(x)$ being an active vertex, and that inactive vertices $v$ with $\ell_G(v) \neq \infty$ have $\NumPathsIn_{G,\ell_G, L'}(v) \geq B^{1/2^s}$
\end{proof}

\subsection{Combining vertex layers}
Now, we show how to combine the vertex layers across different machines from \PartialLayerTree.
We take the minimum layer of a vertex across all the occurrences of that vertex in each of the local pruned trees.
By taking the minimum, the resulting layer assignment is a partial layer assignment.
We present \PartialLayer from \cite{ghaffari2025density} in \cref{alg:mpc:fulllayer}.

\begin{algorithm}
\caption{(\PartialLayer) Assigns layers for entire graph}
\label{alg:mpc:fulllayer}
\begin{algorithmic}[1]
\medskip
\Statex\textbf{Input:} graph $G$, $B, k, L, s, t\in\mathbb{N}$

\Statex\textbf{Output:} partial layer assignment $\ell : V(G) \rightarrow [L] \cup \{\infty\}$
\medskip
\hrule
\medskip
\State $(T_v, \map_v) \gets \Expo\rb{G, k, s, t, B}$ for all $v\in V(G)$
\For{each $v\in V(G)$}
    \State $\ell_v \gets \PartialLayerTree\rb{G, T_v, \map_v,(t + 1)k, L}$
\EndFor
\State $\ell(u) \gets \min\rb{\{\ell_v(x) : v\in V(G), x\in V\rb{T_v}, \map_v(x) = u\}}$ for all $u \in V(G)$
\State \Return $\ell$
\end{algorithmic}
\end{algorithm}

As described in the previous section, we set parameter $a = (t+1)\cdot k$ in \PartialLayerTree because of \cref{theorem:missing}. 
Then, we claim that the partial layer assignment produced by \PartialLayer has out-degree at most $(t+1)\cdot k$.

\begin{lemma}[Modified Claim 3.12 from \cite{ghaffari2025density}]
\label{lemma:orientation:out-degree}
Let $G$ be a graph, $B,k,L,s,t\in\mathbb{N}$, and let $\ell$ be the partial layer assignment returned by $\PartialLayer(G,B,k,L,s,t)$. Then, for every $v\in V(G)$ with $\ell(v) < \infty$, it holds that
\[|\{u\in N_G(v) : \ell(u) \geq \ell(v)\}| \leq (t+1)\cdot k.\]
\end{lemma}
\begin{proof}
We refer the reader to the proof of Claim $3.12$ in \cite{ghaffari2025density}.
\end{proof}

Using \PartialLayer, we compute a partial layer assignment with $L$ layers and out-degree $(t+1)\cdot k$.
However, we also want to have some vertex set size reduction property similar to the analysis of \cref{alg:peeling-orientation}. We assume $k\geq (2+2\eps)\alpha(G)$ and we use the parameters for $L$ and $s$ from \cref{lemma:parameters}.
Let us consider the vertex sets $V_i = \{v \in V(G) : \ell(v) \geq i\}$ using the layers $\ell$ returned by \PartialLayer.
$V_\infty$ consists of all the vertices with layer $\infty$.
Then, we have the following lemma.

\begin{lemma}
\label{lemma:orientation:vertex-size}
Let $G$ be a graph, $B,k,L,s,t\in\mathbb{N}$, and let $\ell$ be the partial layer assignment returned by $\PartialLayer(G,B,k,L,s,t)$. Given
\[k\geq (2+2\eps)\alpha(G)\text{, }L = \left\lfloor\frac{\log^{\frac{t+1}{t+2}}\rb{B}}{4\log(2k)}\right\rfloor\text{, and } s = \left\lceil \frac{\log L}{t + 1}\right\rceil,\]
then $|V_{\infty}| \leq \frac{n}{(1+\eps/2)^L}$.
\end{lemma}
\begin{proof}
Let us consider the full layer assignment $\ell_G : V(G) \rightarrow \left[\lceil\log_{1+\eps} n \rceil\right]$ with out-degree $k$ produced by running \cref{alg:peeling-orientation} on $G$.
Let $L' = 2^{s(t+1)}$.
From the proof of \cref{lemma:parameters}, we have that the number of vertices $v$ with $\NumPathsIn_{G,\ell_G, L'}(v) \geq B^{1/2^s}$ is at most $\frac{n}{2^L}$.
The remaining vertices satisfy $\NumPathsIn_{G,\ell_G, L'}(v) < B^{1/2^s}$.
Using \cref{lemma:orientation:partiallayer}, we know that $\ell(v) \leq \ell_G(v)$ for all $v$ with $\NumPathsIn_{G,\ell_G, L'}(v) < B^{1/2^s}$ and $\ell_G(v) \leq L$.
Now, from the proof of \cref{lemma:orientation:peeling}, the number of vertices with $\ell_G(v) > L$ is at most $\frac{n}{(1+\eps)^L}$.
Therefore, we have that
\[|V_\infty| \leq \frac{n}{(1+\eps)^L} + \frac{n}{2^L} \leq \frac{n}{(1+\eps/2)^L}\]
for sufficiently large $n$.
\end{proof}

\subsection{The algorithm}
Finally, in this section we show how to extend the partial layer assignment produced by \PartialLayer to a full layer assignment where the out-degree of all vertices is bounded by $(t+1)\cdot k$.
We accomplish this by running \PartialLayer multiple times and analyzing how the vertex set size decreases.
We first present an algorithm that can compute $\Theta(\log B)$ layers of the partial layer assignment in \cref{alg:mpc:orientation}.

\begin{algorithm}
\caption{Finds $(2 + \eps)(t+1)$-approximate edge orientation}
\label{alg:mpc:orientation}

\begin{algorithmic}[1]
\medskip
\Statex\textbf{Input:} graph $G$, $\eps\in(0, 1)$, $k,t, B\in\mathbb{N}$

\Statex\textbf{Output:} partial layer assignment $\ell$
\medskip
\hrule
\medskip
\State $L \gets \left\lfloor\frac{\log^{\frac{t+1}{t+2}}\rb{B}}{4\log(2k)}\right\rfloor$, $s \gets\left\lceil \frac{\log L}{t + 1} \right\rceil$, $c \gets \left\lceil \frac{200\log_{1+\eps} B}{L}\right\rceil$
\State $\ell(v) \gets \infty$ for all $v\in V(G)$
\For{$i = 1$ to $c$}
\State $\ell_i \gets \PartialLayer\rb{G, B, k, L, s, t}$
\State $\ell(v) \gets \ell_i(v) + (i - 1)\cdot L$ for all $v\in V(G)$ with $\ell_i(v) \neq \infty$
\State Remove all vertices $v$ with $\ell_i(v) \neq \infty$ from $G$
\EndFor
\State \Return $\ell$
\end{algorithmic}
\end{algorithm}

\begin{lemma}
\label{lemma:orientation:subroutine}
Let $G$ be a graph, $\eps\in(0, 1)$, $k,t, B\in\mathbb{N}$, where $B > k^{2^s}$. Given $k \geq (2+2\eps)\alpha(G)$, \cref{alg:mpc:orientation} satisfies the following. \begin{enumerate}
    \item It returns a partial layer assignment $\ell$ with $L' \geq 200\log_{1+\eps} B$ layers and out-degree $(t+1)\cdot k$.
    \item Let $V_\infty$ be the vertices with layer $\infty$. Then, $|V_\infty| \leq \frac{n}{(1+\eps/2)^{L'}}$.
    \item The algorithm takes $O\rb{\log^{\frac{1}{t+2}}(B)\cdot \log k\cdot \log \log B}$ rounds, $B^{1+o(1)}$ local space, and $O(m) + nB^{1+o(1)}$ global space.
\end{enumerate}
\end{lemma}
\begin{proof}
Using \cref{lemma:orientation:out-degree}, we know that \PartialLayer computes a partial layer assignment with $L$ layers and out-degree $(t+1)\cdot k$.
Therefore, when applying it iteratively $c = \left\lceil \frac{200\log_{1+\eps} B}{L}\right\rceil$ times, we attain a partial layer assignment with $L' = cL \geq 200\log_{1+\eps} B$ layers.

For $V_\infty$, by applying \cref{lemma:orientation:vertex-size} for each iteration of \PartialLayer, we have that
\[|V_\infty| \leq n\cdot \rb{\frac{1}{(1+\eps/2)^L}}^c = \frac{n}{(1+\eps/2)^{L'}}.\]

Finally, the round complexity is $O(c\cdot \log L) = O\rb{\log^{\frac{1}{t+2}}(B)\cdot \log k\cdot \log \log B}$ from calling \PartialLayer $c$ times. The memory complexities come from \cref{theorem:missing}.
\end{proof}

Now, if we use $B = n^\delta$ and increase $c$ slightly, it is not hard to see that \cref{alg:mpc:orientation} can be used to attain a full layer assignment.
However, the global space would be $\tO(n^{1+\delta} + m)$. 
To improve this, we use budget boosting to bring the global space down to $\tO(n + m)$.
The main idea is that we can start with smaller $B$ and as we remove more vertices, we compute their layers and can increase $B$.
We now present the full algorithm.
\begin{theorem}
There is an MPC algorithm that, given a graph $G$, integer $t\geq 0$, and $\eps\in(0,1)$, computes a $(2+\eps)(t+1)$-approximate edge orientation with probability at least $1 - n^{-5}$.
The algorithm requires $O\rb{\log^{1/(t+2)} n \cdot \poly\log\log (n)}$ rounds, $O(n^\delta)$ local space, and $\tO(n + m)$ global space, where $\delta \in (0,1)$ is any fixed constant.
\end{theorem}
\begin{proof}
Let $\cAMPC$ be our MPC algorithm. We first describe $\cAMPC$ and then provide its analysis. \\

\noindent \textbf{Algorithm description:} $\cAMPC$ considers all $1 \leq k' \leq n$ in parallel using powers of $(1+\eps)$. For each parallel instance of $k'$, we do the following.

\textbf{Edge partitioning:} We partition the graph into $H = \lceil \frac{\eps^2 k'}{100\log n} \rceil$ groups. Then, we let $k = \frac{2(1+\eps)^2 k'}{H}$ and calculate an edge orientation with maximum out-degree $k$ for each group.

\textbf{Initial peeling:} For each group in parallel, we first perform $\lceil 100\log_{1+\eps} k\rceil$ iterations of peeling following \cref{alg:peeling-orientation}. Then, we let $B_1 = k^{100}$ and $U_1$ be the remaining vertex set.

\textbf{Budget boosting:} We run $O(\log \log n)$ phases. In phase $i$, we run \cref{alg:mpc:orientation} on the induced subgraph of $U_i$ to attain a partial layer assignment $\ell_i$.
To combine the partial layer assignments across the phases, we add an offset $L_i$ to the layers where
\[L_i = \lceil 100\log_{1+\eps} k\rceil + \sum_{j = 1}^{i - 1} \lceil 400\log_{1+\eps} B_j\rceil.\]
After phase $i$, we boost the budget by setting $B_{i + 1} = \min\rb{B_i^{100}, n^{\delta/2}}$.
Additionally, $U_{i + 1}$ is updated to the new remaining vertex set (vertices that are still in layer $\infty$).

Finally, we combine the edge orientations across all the groups. We consider the smallest $k'$ where the algorithm is able to compute a full layer assignment using $O(\log \log n)$ phases and does not go over the local and global space limits. The edge orientation computed for this $k'$ is our final edge orientation.\\

\noindent \textbf{Algorithm round and memory complexity:} We consider $k'$ that satisfies $\alpha(G) \leq k' \leq (1+\eps)\alpha(G)$.
After edge partitioning, $k = O(\log n)$ and $k$ is at least a $(2+2\eps)$ factor more than the minimum out-degree over all edge orientations of the groups.
Therefore, the memory and round complexity in \cref{theorem:missing} holds.
The initial peeling can be simulated in $\lceil 100\log_{1+\eps} k\rceil = O(\log \log n)$ rounds.
Now,
\[|U_1| \leq \frac{n}{(1+\eps)^{\lceil 100\log_{1+\eps} k\rceil}} \leq \frac{n}{k^{100}}\]
following the proof of \cref{lemma:orientation:peeling}.
Therefore, when using $B_1 = k^{100}$, $|U_1|\cdot B_1 = O(n)$.
Now, for each phase $i$, we run \cref{alg:mpc:orientation} on the induced subgraph of $U_i$.
Using \cref{lemma:orientation:subroutine}, this gives us an additional $\lceil 200\log_{1+\eps} B_i\rceil$ layers.
From the vertex set size drop property in  \cref{lemma:orientation:subroutine}, we have that
\[|U_{i+1}| \leq \frac{|U_i|}{(1+\eps/2)^{\lceil 200\log_{1+\eps} B_i\rceil}} \leq \frac{|U_i|}{B_i^{100}}.\]
Therefore, by setting $B_{i + 1} = \min\rb{B_i^{100}, n^{\delta/2}}$, we maintain that $|U_i|\cdot B_i = O(n)$.
Now, from \cref{lemma:orientation:subroutine}, \cref{alg:mpc:orientation} will take $O\rb{\log^{\frac{1}{t+2}}(B_i)\cdot \log^2 \log n}$ rounds and uses $\tO(B_i)$ local space with $\tO(|U_i|\cdot B_i + m) = \tO(n + m)$ global space.
It is not hard to see that after $j = O(\log \log n)$ total phases, $U_j = \emptyset$ given the conditions on $k$.

Therefore, this gives us a total of $O\rb{\log^{\frac{1}{t+2}}(B_i)\cdot \log^3 \log n}$ rounds.
Additionally, the local space is $O(n^{\delta})$ and the global space is $\tO(n + m)$.\\

\noindent\textbf{Algorithm approximation:} For each group, we compute a layer assignment with out-degree $(t+1)k$ using \cref{lemma:orientation:out-degree}.
From \cref{lemma:edge-partition}, each group has an edge orientation with a maximum out-degree of $k$.
Therefore, when combining the edge orientations, we attain a full edge orientation with maximum out-degree $2(1+\eps)^2(t+1) k'$.
We know there exists $k'$ that satisfies $\alpha(G) \leq k' \leq (1+\eps)\alpha(G)$, so this gives us a $2(1+\eps)^3(t+1)$-approximate edge orientation.
\end{proof}
 
\section{Coloring} \label{sec:coloring}
In this section, we show how to use our edge orientation algorithm to also produce a density-dependent coloring of the graph. Our approach follows that of \cite{ghaffari2025density}, but we provide details for completeness.

Our previous edge orientation algorithm creates a level partition of the vertices of the graph $H_1 \sqcup H_2\sqcup \ldots \sqcup H_L$. If the out-degree of all vertices is bounded by $k$, then we can color all the vertices of the graph using $k + 1$ colors: going in order from the highest layer to the lowest layer, assign each vertex the color that is different from all of its out-neighbors.
The state of the art \LOCAL algorithms from \cite{halldorsson2022near, ghaffari2024near} are able to color the vertices in each layer in $O(\log^{1.67}\log n)$ \LOCAL rounds.

\subsection{Vertex partitioning}
As in \cref{sec:orientation}, we would like to reduce the coloring problem on the input graph to the coloring problem on a collection of graphs, each with arboricity of $O(\log n)$, so that we are able to use our exponentiate and prune algorithm.
Since we are coloring vertices, it is natural that a vertex appears in exactly one of the graphs in a collection.
To achieve this, we partition vertices, rather than edges as in \cref{sec:orientation}, and consider induced subgraph on each partition.
The following standard partitioning lemma, used also in \cite{ghaffari2025density}, gives the required reduction.
\begin{lemma}
\label{lemma:vertex-partition}
Let $G$ be a graph, $k \geq \lceil \alpha(G) \rceil$, and $\eps\in(0, 1)$. Let $H = \lceil \frac{\eps^2 k}{100\log n} \rceil$ and $G_1, G_2,\ldots, G_H$ be the vertex-induced subgraphs from partitioning the vertices into $H$ groups uniformly at random. Then, whp $\max_{i\in[H]} \alpha(G_i) \leq \frac{(1+\eps)k}{H} = O(\log n/\eps^2)$.
\end{lemma}
\begin{proof}
If $H = 1$, $\alpha(G)$ is bounded by $\frac{100\log n}{\eps^2}$.
So, we assume that $H > 1$. 
Fix an orientation $A$ of the edges of $G$, such that the maximum out-degree of $A$ is $k$.
For a vertex $v$, let $G_i$ be the graph containing $v$.
Then, we have that
\[\E{\deg_{G_i}^{out}(v)} \leq \frac{k}{H}\]
which is the expected number of out-neighbors of $v$ within $A$ that are also in $G_i$.
Since vertices are assigned to partitions uniformly at random, $\deg_{G_i}^{out}(v)$ is a sum of independent $0/1$ random variables.
Using the Chernoff bound, we see that
\[\prob{\deg_{G_i}^{out}(v) \geq (1+\eps)k/H} \leq \exp\rb{-\eps^2k/(3H)} \leq n^{-6}.\]
Therefore, taking the union bound over all $v$, we have that each vertex will have at most $(1+\eps)k/H$ out-neighbors in the same group with probability at least $1 - n^{-5}$.
This gives us our desired result.
\end{proof}

\subsection{From vertex partitioning to vertex coloring}
Our vertex coloring algorithm uses edge orientation, so we first discuss how that algorithm affects the coloring.
The approximation/round-complexity tradeoff our edge orientation algorithm achieves is parametrized by $t$. 
In the rest of this section, fix $t$ to be any valid parameter in $O(\log n)$ for which our edge orientation algorithm computes $(2+\eps) \cdot (t+1)$ approximation.
Our goal is to obtain an $O(t\cdot \alpha(G))$ coloring.  
We randomly partition the vertices into $H$ groups as in \cref{lemma:vertex-partition} and color the induced subgraphs using disjoint palettes.  
With high probability, each induced subgraph has arboricity $O(\log n)$.

We now describe how the coloring is performed.
Fix one graph $G_i$ from the vertex partition, and let
\[
    V(G_i) = H_1\sqcup H_2\sqcup \cdots \sqcup H_L
\]
be the layer partition produced by the orientation algorithm. 
Let $d_i$ be the resulting out-degree bound.
Since $\alpha(G_i)=O(\log n)$ with high probability, we have $d_i=O(t\log n)$.

We assign to $G_i$ a palette of $d_i+1$ colors that is disjoint from the palettes assigned to the other parts.
Over all parts, this gives $O(t \cdot \alpha(G))$ colors.
We color the layers of $G_i$ starting from the last layer $H_L$ and proceeding downwards to $H_1$.
Suppose that all layers above $H_j$ have already been colored.  
Then coloring $H_j$ reduces to a degree-plus-one list-coloring instance: each vertex removes from its list the colors already used by its neighbors in higher layers,
and the remaining task is to color the graph induced by $H_j$. 
We use the LOCAL list-coloring algorithms of \cite{halldorsson2022near, ghaffari2024near}, which run in $O(\log^{1.67}\log n)$ LOCAL rounds.

To simulate several layer-coloring steps at once, we use directed graph
exponentiation as in \cite{ghaffari2025density}.
We use the auxiliary digraph in which every cross-layer edge points upward and every same-layer edge is kept in both directions.
Its out-degree is at most $d_i$.
Assume that all layers above $H_j$ have already been colored, and consider a block of $b$ consecutive lower layers.
To color this block locally, it suffices to gather directed neighborhoods of radius
\[
    O(b\log^{1.67}\log n),
\]
since the $O(\log^{1.67}\log n)$-round LOCAL coloring routine is applied once per layer, starting from the highest layer of the block and moving downward.

As the directed out-degree is $d_i = O(t\log n)$ and $t = O(\log n)$, such a
neighborhood has size
\[
    d_i^{O(b\log^{1.67}\log n)} = 2^{O(b\log^{2.67}\log n)}.
\]
Thus we choose
\[
    b=\Theta\rb{\frac{\log n}{\log^{2.67}\log n}}
\]
with a sufficiently small constant, so that these neighborhoods fit in
local memory.
Processing the blocks from higher layers to lower layers colors all vertices with only a $\poly(\log\log n)$ MPC-round overhead.

Thus, the coloring step preserves the asymptotic round-complexity tradeoff
of the orientation algorithm and uses $(2+\eps)(t + 1)\alpha(G) + 1$ colors.

\section{Densest subgraph} \label{sec:densest-subgraph}
In this section, we demonstrate how to apply our exponentiate-and-prune algorithm to find an approximate densest subgraph.

\subsection{Base algorithm}
Our main goal is to simulate the fixed threshold peeling algorithm from \cite{mitrovic2026new}, which we present in \cref{alg:peeling-DS}.

\begin{algorithm}
\caption{(\dsPeeling) Finds $(2+\eps)$-approximate densest subgraph}
\label{alg:peeling-DS}
\begin{algorithmic}[1]
\medskip
\Statex\textbf{Input:} graph $G$, parameters $\eps$, $k$
\Statex\textbf{Output:} a subset of vertices of $G$
\medskip
\hrule
\medskip
\State $\ell(v) \gets \infty$ for all $v\in V(G)$
\For{$i = 1$ to $\lceil\log_{1+\eps} n\rceil$}

    \State $A\gets$ all vertices $v\in V(G)$ with $d_{G}(v) < k$

    \If{$|A| \leq \frac{\eps}{1+\eps}|V(G)|$}
        \Return $V(G)$
    \EndIf

    \For{each $v\in A$}
        $\ell(v)\gets i$
    \EndFor
    
    \State Remove vertices in $A$ from $G$
\EndFor
\end{algorithmic}
\end{algorithm}

\cref{alg:peeling-DS} assigns layers to vertices, and the layers have the following properties.
\begin{lemma}
\label{lemma:DS-1}
For all $v\in V(G)$ with $1 < \ell(v) < \infty$, it holds that
\[|\{u\in N_G(v) : \ell(u) \geq \ell(v) - 1\}| \geq k.\]
\end{lemma}
\begin{proof}
Suppose that $|\{u\in N_G(v) : \ell(u) \geq \ell(v) - 1\}| < k$. This implies that at iteration $\ell(v) - 1$, the degree of $v$ was less than $k$.
Therefore, the algorithm would peel $v$ and its layer would be $\ell(v) - 1$ instead of $\ell(v)$.
This is a contradiction.
\end{proof}

\begin{lemma}
\label{lemma:DS-2}
Let $S^*$ be a densest subgraph of $G$. If $k \leq \rho^*(G)$, then $\ell(v) = \infty$ for all $v \in S^*$.
\end{lemma}
\begin{proof}
Let $G^*$ be the induced subgraph of $S^*$.
First, we prove that if $S^*$ is a densest subgraph, then the degree of all the vertices in $G^*$ must be at least $\rho^*(G)$.
Consider any $v\in S^*$.
Then, because $S^*$ is a densest subgraph, removing $v$ from the subgraph $G^*$ cannot increase the density.
So, we have
\[\frac{|E(G^*)|}{|S^*|} \geq \frac{|E(G^*)| - \deg_{G^*}(v)}{|S^*| - 1} \implies \deg_{G^*}(v) \geq \frac{|E(G^*)|}{|S^*|} = \rho^*(G).\]
Using this, we have that none of the vertices in $S^*$ can be peeled when $k\leq \rho^*(G)$ and, therefore, they all have layer $\infty$.
\end{proof}

Using these properties, we attain the following approximation result about \cref{alg:peeling-DS}.

\begin{lemma}
Let $G$ be a graph, $\eps\in(0, 1)$, and $k\leq \rho^*(G)$. Then, \cref{alg:peeling-DS} returns a subgraph with density at least $\frac{k}{2(1+\eps)}$.
\end{lemma}
\begin{proof}
First, we observe that for every iteration that \cref{alg:peeling-DS} does not return, the vertex set of $G$ decreases by at least a factor of $(1+\eps)$.
So, after more than $\log_{1+\eps} n$ iterations, the graph will become empty.
However, \cref{lemma:DS-2} shows that the graph will not become empty from peeling when $k\leq \rho^*(G)$, so the algorithm will return a subgraph.

Now, to lower bound the density of the subgraph returned, we loosely use the property in \cref{lemma:DS-1}. Let $S$ be the returned subgraph. When $S$ is returned, at least a $1/(1+\eps)$ fraction of its vertices have degree at least $k$. This gives us that
\[\rho(S) \geq \frac{\frac{k}{2(1+\eps)}\cdot |S|}{|S|} = \frac{k}{2(1+\eps)}\]
which is our desired density lower bound.
\end{proof}

Therefore, if $k$ is chosen close enough to $\rho^*(G)$, \cref{alg:peeling-DS} computes a $(2+\eps)$-approximation of the densest subgraph.

\subsection{Preprocessing}
\label{sec:ds-sampling}
In order to properly use our exponentiate and prune algorithm, we need to sample the graph so that an approximate densest subgraph is preserved and the arboricity of the graph is $O(\log n)$.
Our uniform sampling will consist of selecting each edge of the graph independently with probability $p = \Theta\rb{\frac{\log n}{D^*}}$ where $D^*$ is the density of the densest subgraph.
This leads us to the following lemma.
\begin{lemma}
\label{lemma:sampling-DS}
Let $G$ be a graph, $D^* = \rho^*(G)$, and $\eps \in(0, 1/2)$. For a given $c \geq 1$, let $p = \min\rb{1, \frac{20c\log n}{\eps^2 D^*}}$. Let $H$ be a graph obtained from sampling each edge of $G$ independently with probability $p$. Then, whp the following hold:
\begin{itemize}
    \item[(A)] Let $S$ be a subgraph of $G$ with density at least $D^*/c$. Then, the density of $S$ in $H$ is in $[(1-\eps)p\rho_G(S), (1+\eps)p\rho_G(S)]$.

    \item[(B)] Let $S$ be a subgraph of $G$ with density at most $(1-2\eps)D^*/c$. Then, the density of $S$ in $H$ is at most $(1-\eps)pD^*/c$.
\end{itemize}
\end{lemma}
\begin{proof}
If $p=1$, everything holds. Therefore, we assume $p = \frac{20c\log n}{\eps^2 D^*}$. \\

\noindent \textbf{Proof of (A).} Fix any integer $k$ with $1\leq k\leq n$. Consider any subgraph $S$ of $G$ with exactly $k$ vertices and density at least $D^*/c$.
Then,
\[\E{|E_H(S)|} = pk\rho_G(S) \geq pkD^*/c.\]
Therefore, using the Chernoff bound, we have that
\[\prob{||E_H(S)| - pk\rho_G(S)| \geq \eps pk\rho_G(S)} \leq 2\exp\rb{-\eps^2pkD^*/(3c)} \leq n^{-6k}.\]
This implies that the density of $S$ in $H$ is in $[(1-\eps)p\rho_G(S), (1+\eps)p\rho_G(S)]$ with probability at least $1 - n^{-6k}$.
Taking the union bound over all possible $k$ and subgraphs $S$, we have that (A) holds with probability at least $1 - n^{-4}$. \\

\noindent \textbf{Proof of (B).} Fix any integer $k$ with $1\leq k\leq n$. Consider any subgraph $S$ of $G$ with exactly $k$ vertices and density at most $(1-2\eps)D^*/c$. Then,
\[\E{|E_H(S)|} = pk\rho_G(S) \leq (1-2\eps)pkD^*/c.\]
Therefore, using the Chernoff bound, we have that
\[\prob{|E_H(S)| - (1-2\eps)pkD^*/c \geq \eps pkD^*/c} \leq \exp\rb{-\eps^2pkD^*/(3c)} \leq n^{-6k}.\]
This implies that the density of $S$ in $H$ is at most $(1-\eps)pD^*/c$ with probability at least $1 - n^{-6k}$.
Taking the union bound over all possible $k$ and subgraphs $S$, we have that (B) holds with probability at least $1 - n^{-4}$. 
\end{proof}

In other words, a constant approximation of the densest subgraph in the sampled graph reflects a similar constant approximation of the densest subgraph in the original graph.

\subsection{Combining vertex layers}\label{sec:ds:combining}
Now, we want our exponentiate and prune algorithm to produce approximately the same layers as \cref{alg:peeling-DS} in order to compute an approximate densest subgraph in sub-linear MPC.
For our analysis to work, we take the \textit{maximum} over all layer assignments from various trees for each vertex $v\in V(G)$ instead of the minimum that edge orientation uses.
It is important to note that our resulting function from combining the partial layer assignments of the different trees is \textit{not} necessarily a partial layer assignment because we are taking the maximum instead of the minimum.
We present the algorithm in \cref{alg:mpc:fulllayer-DS}.

\begin{algorithm}
\caption{(\PartialLayerDS) Assigns layers for entire graph}
\label{alg:mpc:fulllayer-DS}

\begin{algorithmic}[1]
\medskip
\Statex\textbf{Input:} graph $G$, $B, k, L, s, t\in\mathbb{N}$

\Statex\textbf{Output:} function $\ell : V(G) \rightarrow [L] \cup \{\infty\}$
\medskip
\hrule
\medskip
\State $(T_v, \map_v) \gets \Expo\rb{G, k, s, t, B}$ for all $v\in V(G)$
\For{each $v\in V(G)$}
    \State $\ell_v \gets \PartialLayerTree\rb{G, T_v, \map_v,(t + 1)k, L}$
    \For{each $x \in V(T_v)$}
        \State $d_{T_v}(x)\gets$ depth of $x$ in $T_v$
        \If{$\ell_v(x) = \infty$ and $2^{s(t+1)} - d_{T_v}(x) \leq L - 1$ and $\map_v(x)$ is active}
            \State $\ell_v(x) = 2^{s(t+1)} - d_{T_v}(x) + 1$
        \EndIf
    \EndFor
\EndFor
\For{each $u\in V(G)$}
    \State $\ell(u) \gets \max\rb{\{\ell_v(x) : \text{active }v\in V(G), x\in V\rb{T_v}, \map_v(x) = u\}}$
\EndFor
\State \Return $\ell$
\end{algorithmic}
\end{algorithm}

To properly implement taking the maximum over all partial layer assignments of various trees, we have to be careful of vertices with layer $\infty$.
If we naively take the maximum, it is possible for many vertices to have layer $\infty$ that should not due to how the local pruned tree views only have height $2^{s(t+1)}$.
To fix this, we decrease the layer of vertices if their height $h$ in the local pruned tree is too low.
This provides a more accurate certification of the layer of the vertex because the height of the vertex is an upper bound on the layer of a vertex, if it is not $\infty$.
Therefore, if the layer is $\infty$, we want to set the layer of the vertex to be $h + 1$.
This is reflected in lines $4$ through $7$ in \cref{alg:mpc:fulllayer-DS}.
Note that we compute the "height" of a vertex using the depth of the vertex instead.
In the analysis below, we want the height of children vertices to be exactly one less than the height of the parent, so we define the height in this way.

We now prove the following lemmas about \PartialLayerDS that reflect the properties in \cref{lemma:DS-1} and \cref{lemma:DS-2} from \cref{alg:peeling-DS}.
\begin{lemma}
\label{lemma:ds:k-lowerbound}
Let $G$ be a graph, $B, k, L, s, t\in \mathbb{N}$, and $\ell: V(G) \rightarrow [L] \cup \{\infty\}$ be the function returned by $\PartialLayerDS(G, B, k, L, s, t)$. Then, for every active $v\in V(G)$ with $1 < \ell(v) \leq L$, it holds that
\[|\{u\in N_G(v) : \ell(u) \geq \ell(v) - 1\}| \geq k,\]
and for every active $v\in V(G)$ with $\ell(v) = \infty$, it holds that
\[|\{u\in N_G(v) : \ell(u) \geq L\}| \geq k.\]
\end{lemma}
\begin{proof}
Consider any active vertex $v\in V(G)$ with $1 < \ell(v) < \infty$.
Then, we observe from line $9$ of \cref{alg:mpc:fulllayer-DS} that there exists some active vertex $w\in V(G)$ and $x\in V\rb{T_w}$ such that $\map_w(x) = v$ and $\ell_w(x) = \ell(v)$.
Now, we have two cases: either $\ell_w(x)$ was originally $\infty$ and was set to $\ell(v)$ or it stayed as $\ell(v)$.
We let $\ell'_w$ be the original layers of vertices in $T_w$ (without lines $4$ through $7$ of \cref{alg:mpc:fulllayer-DS}).

In the first case where $\ell'_w(x)$ is $\infty$ and $\ell_w(x) = \ell(v)$ is not anymore, we must have that $\ell(v) = 2^{s(t+1)} - d_{T_w}(x) + 1$.
Now, since $\ell'_w(x)$ is $\infty$, we see that $|\children_{T_w}(x)| + \Missing(x) \geq (t+1)k$ during iteration $L$ of \PartialLayerTree.
Noting $v$ is active, we know that $\Missing(x) \leq tk$ and so $|\children_{T_w}(x)| \geq k$ using \cref{theorem:missing}.
Therefore, $x$ has at least $k$ children with original layer $L$ or $\infty$.
Let $\mathcal{C}$ be the set of these children.
Consider any child $c\in\mathcal{C}$ with $\ell'_w(c) = \infty$.
We know that its layer will be set to $2^{s(t+1)} - d_{T_w}(c) + 1 = \ell(v) - 1$ in the tree if $\map_w(c)$ is active, or it will stay at $\infty$.
Now, because line $9$ takes the maximum over all the layers of the various trees, we must have that all $c\in\mathcal{C}$ have $\ell(\map_w(c)) \geq \ell(v) - 1$.

In the second case where $\ell'_w(x) = \ell_w(x) = \ell(v)$, we observe the peeling performed by \PartialLayerTree.
We must have that $x$ was removed during iteration $\ell(v)$ and was not removed during iteration $\ell(v) - 1$.
In order for $x$ to not be removed during iteration $\ell(v) - 1$, we must have that $|\children_{T_w}(x)| \geq k$ using \cref{theorem:missing}.
This means that $x$ has at least $k$ children during iteration $\ell(v) - 1$.
Let $\mathcal{C}$ be the set of these children.
For any $c\in\mathcal{C}$, we must have that $\ell_w'(c) \geq \ell(v) - 1$.
Now, we consider any child $c\in\mathcal{C}$ with $\ell'_w(c) = \infty$.
Since $\ell'_w(x) = \ell(v)$, the height of $x$ must be at least $\ell(v)$.
Therefore, the layer of $c$ will be set to at least $\ell(v) - 1$ in the tree if $\map_w(c)$ is active, or it will stay at $\infty$.
Now, because line $9$ takes the maximum over all the layers of the various trees, we must have that all $c\in\mathcal{C}$ have $\ell(\map_w(c)) \geq \ell(v) - 1$.

Therefore, in both cases, we attain our desired result. \\

Now, we consider any active vertex $v\in V(G)$ with $\ell(v) = \infty$.
We observe from line $9$ of \cref{alg:mpc:fulllayer-DS} that there exists some active vertex $w\in V(G)$ and $x\in V(T_w)$ such that $\map_w(x) = v$ and $\ell_w(x) = \infty$.
Once again, we let $\ell'_w$ be the original layers of vertices in $T_w$.
Since $\ell_w(x) = \infty$, we must have that $|\children_{T_w}(x)| \geq k$ during iteration $L$ of \PartialLayerTree using \cref{theorem:missing}.
This means that $x$ has at least $k$ children during iteration $L$.
Let $\mathcal{C}$ be the set of these children.
Then, $\ell'_w(c) = L$ or $\ell'_w(c) = \infty$ for all $c\in\mathcal{C}$.
Now, we consider any $c\in\mathcal{C}$ with $\ell'_w(c) = \infty$ and $\ell_w(c) \neq \infty$.
Since $\ell_w(x) = \infty$, we must have that the height of $x$ was greater than $L - 1$ while the height of $c$ is exactly $L - 1$.
This implies that $\ell_w(c) = L$.
Therefore, because line $9$ takes the maximum over all the layers of the various trees, we must have that all $c\in\mathcal{C}$ have $\ell(\map_w(c)) \geq L$
\end{proof}

\begin{lemma}
\label{lemma:ds:inf-level}
Let $G$ be a graph, $B, k, L, s, t\in \mathbb{N}$. Let $k\leq \frac{\rho^*(G)}{t + 1}$ and $\ell: V(G) \rightarrow [L] \cup \{\infty\}$ be the function returned by $\PartialLayerDS(G, B, k, L, s, t)$.
Let $S^*$ be a densest subgraph of $G$ and $2^{s(t+1)} \geq L$. Then, $\ell(v) = \infty$ for all $v\in S^*$.
\end{lemma}
\begin{proof}
Let $G^*$ be the induced subgraph of $S^*$.
Then, we have that $\deg_{G^*}(v)\geq (t+1)k$ for all $v\in S^*$.

Now, for all $v\in S^*$, let us consider the tree $T_v$ with root $r$ where $\map_v(r) = v$.
We know that \PartialLayerTree will only remove vertex $r$ if $|\children_{T_v}(r)| + \Missing(r) < (t + 1)k$.
However, since $\deg_{G^*}(v) \geq (t+1)k$ for all $v\in S^*$, it will never be removed.
Therefore, $\ell_v(r) = \infty$ since its height will be $2^{s(t+1)} \geq L$.
Therefore $\ell(v) = \infty$ for all $v\in S^*$ from line $9$ of \cref{alg:mpc:fulllayer-DS}.
\end{proof}

Finally, we want to upper bound the layers of vertices with $\NumPathsIn_{G,\ell_G, L'}(v) < B^{1/2^s}$ where $L' = 2^{s(t+1)}$ and $\ell_G:V(G) \rightarrow [L] \cup \{\infty\}$ is any partial layer assignment with out-degree $k$.
This is similar to the property of \cref{lemma:orientation:partiallayer} in \cref{sec:orientation} and shows that the layers produced by \PartialLayerDS will be a subset of the layers produced by any partial layer assignment $\ell_G$ for the majority of vertices.

\begin{lemma}
\label{lemma:ds:nested}
Let $G$ be a graph, $k, L\in\mathbb{N}$, and $\ell_G:V(G)\rightarrow [L]\cup \{\infty\}$ be a partial layer assignment with out-degree $k$.
Let $B, s, t\in\mathbb{N}$ be parameters that satisfy $s(t+1) \geq \log L$.
Let $\ell$ be the function returned by $\PartialLayerDS(G, B, k, L, s, t)$.
Consider any $v\in V(G)$ with $\ell_G(v) \neq \infty$ that satisfies $\NumPathsIn_{G,\ell_G, L'}(v) < B^{1/2^s}$ where $L' = 2^{s(t+1)}$.
Then, $\ell(v) \leq \ell_G(v)$.
\end{lemma}
\begin{proof}
Consider any active $w \in V(G)$ and any $x\in V(T_w)$ where $\map_w(x) = v$.
If the height of $x$ in $T_w$ is at least $\ell_G(v)$, then we have that $\ell_w(x) \leq \ell_G(v)$ using \cref{lemma:orientation:partiallayer}.
Otherwise, we have that the height of $x$ in $T_w$ is less than $\ell_G(v)$. 
The only way for $\ell_w(x)$ to be more than $\ell_G(v)$ is if $\ell_w(x) = \infty$.
However, this is impossible because \PartialLayerDS will set $\ell_w(x)$ to be its height plus $1$ which is at most $\ell_G(v)$.
Therefore, in either case, $\ell_w(x) \leq \ell_G(v)$ for any active $w \in V(G)$ and any $x\in V(T_w)$ where $\map_w(x) = v$.
Taking the maximum over all these layers gives us $\ell(v) \leq \ell_G(v)$.
\end{proof}

Therefore, the layers produced by \PartialLayerDS approximate the peeling performed by \dsPeeling. Specifically, we can look at the vertex sets of the layers and apply the stopping condition in \dsPeeling by comparing the sizes of vertex sets.
\cref{lemma:ds:k-lowerbound} and \cref{lemma:ds:inf-level} give us the tools necessary to prove that stopping when the vertex set sizes do not decrease enough results in an approximate densest subgraph.

\subsection{Approximate densest subgraph algorithm for $t \leq 2$}
\label{sec:DS-alg}
We now describe how to compute an approximate densest subgraph from the layers returned by \PartialLayerDS for $t \leq 2$.
Specifically, we compute a $(4+\eps)$-approximate densest subgraph in $\tO\rb{\log^{1/3} n}$ rounds and a $(6+\eps)$-approximate densest subgraph in $\tO\rb{\log^{1/4} n}$ rounds.
In the previous section, we presented tools that allow us to compute an approximate densest subgraph or observe a significant decrease in the vertex set size.
However, this is only true if we are not taking the inactive vertices into account.
The number of inactive vertices can be large, unfortunately, and do not satisfy the properties from the previous section.
This can cause our vertex set size comparisons to be diluted by inactive vertices that are never peeled, resulting in the vertex set size not decreasing by much.
To fix this, we claim that if there are many inactive vertices, we can compute an approximate densest subgraph using a subroutine described below.

Additionally, we have to be careful about the vertices with layer $\infty$.
It is possible that \PartialLayerDS may not assign exactly $L$ layers because of the way it combines the layers across different trees.
As a result, we can only use the stopping conditions from \dsPeeling for the layers it does assign.
Then, there can be a situation where the number of vertices with layer $\infty$ is large but none of the intermediate vertex sets satisfied the stopping conditions.
This does not give us a desired vertex set size decrease.
However, similar to the situation where the number of inactive vertices is large, we also claim that if the number of vertices with layer $\infty$ is large, we can compute an approximate densest subgraph using the subroutine described below.

Our algorithm will use the algorithm \dsSqrtPeeling, which we modify from \cite{mitrovic2026new}, as a subroutine.
The goal of \dsSqrtPeeling is to simulate $L$ iterations of \dsPeeling using $\tO(L/\sqrt{\log n})$ sub-linear MPC rounds.
It takes the following parameters: a graph $G$, approximation parameter $\eps$, memory parameter $\delta$, iteration parameter $c$, and peeling parameter $k$.
For details on \dsSqrtPeeling, we refer the reader to \cref{appendix:sqrt-DS}.
Importantly, it has the following main properties:
\begin{enumerate}
    \item \textbf{Sublinear MPC simulation:} (\cref{lemma:sqrt:memory-round}) The algorithm performs $c$ phases of $\Theta(\sqrt{\log n})$ approximate iterations of peeling, with each phase taking $O(\log \log n)$ rounds. The algorithm requires $O(c\cdot\log\log n)$ total rounds, $O(n^\delta)$ local space, and $O(n^{1+\delta} + m)$ global space.

    \item \textbf{Approximation guarantee:} (\cref{lemma:sqrt:approx}) When the algorithm returns a subgraph, it has density at least $\frac{k}{2(1+\eps)}$.

    \item \textbf{Vertex set decrease:} (\cref{lemma:sqrt:vertex-set-decrease-outer}) After each phase of the algorithm, the vertex set drops by a factor of $(1+\eps)^{\Theta(\sqrt{\log n})}$.

    \item \textbf{$k$-core maintained:} (\cref{lemma:sqrt:nonempty-kcore}) The algorithm never removes any of the vertices in the $k$-core of $G$.
\end{enumerate}

We now present our algorithm that uses both \PartialLayerDS and \dsSqrtPeeling in \cref{alg:mpc:DS}.

\begin{algorithm}
\caption{Finds $(2 + \eps)(t+1)$-approximate densest subgraph for $t\leq 2$}
\label{alg:mpc:DS}

\begin{algorithmic}[1]
\medskip
\Statex\textbf{Input:} graph $G$, $\eps\in(0, 1)$, $\delta\in(0,1)$, $k,t\in\mathbb{N}$

\Statex\textbf{Output:} a subset of vertices of $G$
\medskip
\hrule
\medskip
\State $L \gets \left\lfloor\frac{\delta\log^{\frac{t+1}{t+2}}\rb{n}}{4\log(2k)}\right\rfloor$, $s \gets\left\lceil \frac{\log L}{t + 1} \right\rceil$
\State $\ell \gets \PartialLayerDS\rb{G, n^\delta, k, L, s, t}$
\State $I\gets\{v\in V(G) : v\text{ is inactive}\}$
\State $V_{\infty} \gets \{v\in V(G) : \ell(v) = \infty\}$
\If{$|I| \geq \frac{n}{(1+\eps)^L}$ or $|V_{\infty}| \geq \frac{n}{(1+\eps/4)^L}$}
\State $c \gets \left\lceil\frac{6L}{\sqrt{\delta\log_{1+\eps}n}}\right\rceil$
\State \Return $\dsSqrtPeeling(G, \eps, \delta, c, k)$
\EndIf
\State $V_i \gets \{v\in V(G) : \ell(v) \geq i\}$ for all $i = 1,\ldots, L + 1$
\For{$i = 1\text{ to }L$}
\If{$V_i = V_{\infty}$}
\textbf{break}
\EndIf
\If{$|V_{i+1}| - |I| \geq |V_i|/(1+\eps)$}
\Return $V_i$
\EndIf
\EndFor
\State\Return $V_{\infty}$
\end{algorithmic}
\end{algorithm}

\cref{alg:mpc:DS} uses the parameters for $L, s$ from \cref{lemma:parameters} with $B = n^\delta$.
Additionally, \cref{alg:mpc:DS} defines set $I$ as the set of inactive vertices and $V_{\infty}$ as the set of vertices with layer $\infty$.
Because $I$ could include vertices that should be peeled at iteration $i$ but are not, we remove them from the vertex set $V_{i+1}$ when making a comparison to the size of the vertex set $V_i$.
This is reflected in line $11$ of the algorithm.
Additionally, the situation where the number of inactive vertices or the number of vertices with layer $\infty$ is large is reflected in lines $5$ through $7$.
We claim that running our subroutine \dsSqrtPeeling with $O(L/\sqrt{\log n})$ phases is enough to compute an approximate densest subgraph.
Now, we show in \cref{lemma:DS-converge} that the vertex size still drops significantly, even with the inactive vertices, or we compute an approximate densest subgraph on line $7$, from \dsSqrtPeeling, or line $11$, from our vertex set size comparison.

\begin{lemma}
\label{lemma:DS-converge}
Let $G$ be a graph, $\delta\in(0,1)$, $\eps\in(0,1)$, $k, t\in \mathbb{N}$.
Let $k \leq \frac{\rho^*(G)}{t+1}$. 
Let $S$ be the output of \cref{alg:mpc:DS} ran on these parameters.
Then, it holds that:
\begin{itemize}
    \item \textbf{Good approximation.} $S$ is a subgraph with density at least $\frac{k}{2(1+\eps)}$, or
    \item \textbf{Size reduction.} $S$ contains a densest subgraph of $G$ and $|S| \leq \frac{n}{(1+\eps/4)^L}$.
\end{itemize}
\end{lemma}
\begin{proof}
\cref{alg:mpc:DS} either returns early on lines $7$ or $11$, or it returns at the end on line $12$. We consider these two situations separately. \\

\noindent\textbf{\cref{alg:mpc:DS} returns early on lines $7$ or $11$:}

\noindent If \cref{alg:mpc:DS} returns early on lines $7$ or $11$, we claim that the returned subgraph has density at least $\frac{k}{2(1+\eps)}$.
Let us first consider when it returns on line $7$.
This only happens when $|I| \geq \frac{n}{(1+\eps)^L}$ or $|V_\infty| \geq \frac{n}{(1+\eps/4)^L}$.

First, we consider $|I| \geq \frac{n}{(1+\eps)^L}$.
Let us consider the partial layer assignment $\ell' : V(G) \rightarrow [L'] \cup \{\infty\}$ which is produced by iteratively peeling all vertices with degree less than $k$ until we reach the $k$-core of $G$ (so $L'$ is the number of iterations of peeling needed to reach the $k$-core).
Then, using \cref{lemma:parameters}, we know that the number of inactive vertices with $\ell'(v) < \infty$ is at most $\frac{n}{2^L}$.
This means that at least $\frac{n}{(1+\eps)^L} - \frac{n}{2^L} \geq \frac{n}{2(1+\eps)^L}$ vertices with $\ell'(v) = \infty$ are inactive.
Therefore, we know the $k$-core of $G$ has at least $\frac{n}{2(1+\eps)^L}$ vertices.
With this, \cref{alg:mpc:DS} runs \dsSqrtPeeling with $c = \left\lceil\frac{6L}{\sqrt{\delta\log_{1+\eps}n}}\right\rceil$.
Let us assume, for the sake of contradiction, that \dsSqrtPeeling does not return a subgraph.
Then, using \cref{lemma:sqrt:vertex-set-decrease-outer}, the size of the vertex set at the end is
\[n\cdot \rb{\frac{1}{(1+\eps/2)^{\sqrt{\delta\log_{1+\eps}n}/2}}}^c \leq \frac{n}{(1+\eps/2)^{3L}} < \frac{n}{2(1+\eps)^L}\]
for sufficiently large $n$.
However, the vertex size is less than the size of the $k$-core of $G$ and using \cref{lemma:sqrt:nonempty-kcore}, we know that \dsSqrtPeeling never removes any of the vertices of the $k$-core.
Therefore, we have a contradiction and \dsSqrtPeeling must return a subgraph which has density at least $\frac{k}{2(1+\eps)}$ using \cref{lemma:sqrt:approx}.

Second, we consider $|V_{\infty}| \geq \frac{n}{(1+\eps/4)^L}$.
Let us look at the partial layer assignment $\ell_{\text{sqrt}} : V(G) \rightarrow [L] \cup \{\infty\}$ produced by the peeling of \dsSqrtPeeling on line $7$ with $c = \left\lceil\frac{6L}{\sqrt{\delta\log_{1+\eps}n}}\right\rceil$.
We define $V_{i}^{\text{sqrt}} = \{v\in V(G): \ell_{\text{sqrt}}(v) \geq i\}$.
Let $L' = 2^{s(t+1)} \geq L$.
From \cref{lemma:parameters}, we know that the number of vertices $v\in V(G)$ with $\ell_{\text{sqrt}}(v) \neq \infty$ and $\NumPathsIn_{G,\ell_{\text{sqrt}}, L'}(v) \geq B^{1/2^s}$ is at most $\frac{n}{2^{L}}$.
Then we have that
\[|V_{\infty}^{\text{sqrt}}| \geq |V_\infty| - \frac{n}{2^L} \geq \frac{n}{(1+\eps/4)^L} - \frac{n}{2^L} > \frac{n}{(1+\eps/2)^L}\]
using \cref{lemma:ds:nested} and the previous upper bound on the number of vertices with large $\NumPathsIn_{G,\ell_{\text{sqrt}}, L'}(v)$, assuming sufficiently large $n$.
Then, running \dsSqrtPeeling with $c = \left\lceil\frac{6L}{\sqrt{\delta\log_{1+\eps}n}}\right\rceil$ must return a subgraph using \cref{lemma:sqrt:vertex-sets-decrease} because $|V_{\infty}^{\text{sqrt}}|$ is larger than the vertex size drop of $L$ layers without returning.
Therefore, the returned subgraph has density at least $\frac{k}{2(1+\eps)}$ using \cref{lemma:sqrt:approx}.

Now, let us consider when \cref{alg:mpc:DS} returns on line $11$.
Using \cref{lemma:ds:k-lowerbound}, line $11$ tells us that at least $|V_i|/(1+\eps)$ vertices in $V_i$ have degree at least $k$ within the induced subgraph of $V_i$.
Therefore, $\rho(V_i) \geq \frac{k}{2(1+\eps)}$. \\

\noindent\textbf{\cref{alg:mpc:DS} returns at the end on line $12$:}

\noindent If \cref{alg:mpc:DS} never returns early, then we return $V_\infty$ at the end.
Because of line $5$, $|V_\infty| < \frac{n}{(1+\eps/4)^L}$.
Also, using \cref{lemma:ds:inf-level}, we have that the densest subgraph will be a subset of $V_\infty$.
\end{proof}

Finally, using \cref{alg:mpc:DS}, we present a full algorithm that computes a $(2+\eps)(t+1)$-approximate densest subgraph for $t\leq 2$.

\begin{theorem}
\label{theorem:DS}
There is an MPC algorithm that, given a graph $G$, positive integer $t\leq 2$ and $\eps\in(0,1)$, computes a $(2+\eps)(t+1)$-approximate densest subgraph with probability at least $1 - n^{-4}$.
The algorithm requires $O\rb{\log^{1/(t+2)} n \cdot \poly\log\log (n)}$ rounds, $O(n^\delta)$ local space, and $\tO(n^{1+\delta} + m)$ global space, where $\delta \in (0,1)$ is any fixed constant.
\end{theorem}
\begin{proof}
Let $\cAMPC$ be our MPC algorithm. We first describe $\cAMPC$ and then provide its analysis. \\

\noindent \textbf{Algorithm description:} $\cAMPC$ considers all $1 \leq k \leq n$ in parallel using powers of $(1+\eps)$. For each parallel instance of $k$,
\begin{enumerate}
    \item The algorithm first obtains a graph $H$ by sampling each edge of $G$ independently with probability $\min\left(1, \frac{160\log n}{\eps^2 k}\right)$.

    \item The algorithm invokes \cref{alg:mpc:DS} on parameters $H$, $\eps$, $\delta$, $(1-\eps)pk$, and $t$. If \cref{alg:mpc:DS} returns early, we call the returned vertex set $S$ a \textit{potential densest subgraph} and end this parallel instance.

    \item If \cref{alg:mpc:DS} does not return early, we let $H'$ be the induced subgraph of the returned vertex set $S$. Then, we repeat the algorithm starting from step $2$ using graph $H'$ instead of $H$.
\end{enumerate}
We allow the algorithm to invoke \cref{alg:mpc:DS} a maximum of $\frac{4\log_{1+\eps} n}{L}$ times where $L =  \left\lfloor\frac{\delta\log^{\frac{t+1}{t+2}}\rb{n}}{4\log(2k)}\right\rfloor$. 
Then, $\cAMPC$ returns the potential densest subgraph with the largest density over all instances of $k$ as the final approximate densest subgraph. \\

\noindent \textbf{Algorithm round and memory complexity:} $\cAMPC$ invokes \cref{alg:mpc:DS} a maximum of $\frac{4\log_{1+\eps} n}{L}$ times.
For each invocation, \cref{alg:mpc:DS} runs \PartialLayerDS which runs \Expo and performs local computation on each machine with $O(1)$ rounds of communication to combine vertex layers.
Therefore, from \cref{theorem:missing}, the algorithm uses $O(s\cdot (t + 1)) = O(\log \log n)$ rounds from \Expo.
However, it is possible for \cref{alg:mpc:DS} to return early by running \dsSqrtPeeling.
From \cref{lemma:sqrt:memory-round}, \cref{alg:mpc:DS} will use $O\rb{\frac{L}{\sqrt{\log n}}\cdot \log \log n} = O\rb{\log^{1/(t+2)} n \cdot \log \log n}$ rounds for $t \leq 2$ and then $\cAMPC$ will stop.
Since \cref{alg:mpc:DS} is only invoked at most $\frac{4\log_{1+\eps} n}{L} = O\rb{\log^{1/(t+2)}n \cdot \log \log n}$ times, we have a final round complexity of $O\rb{\log^{1/(t+2)} n \cdot \poly\log\log (n)}$.

All of the algorithms mentioned use $O(n^\delta)$ local space and $\tO(n^{1+\delta} + m)$ global space from \cref{theorem:missing} and \cref{lemma:sqrt:memory-round}. \\

\noindent \textbf{Algorithm approximation:} We observe that there exists a value of $k$ such that
\[\frac{\rho^*(G)}{(1+\eps)(t + 1)}\leq k \leq \frac{\rho^*(G)}{t + 1}.\]
Using \cref{lemma:sampling-DS}, we have that
\[\frac{(1-\eps)\rho^*(H)}{(1+\eps)^2(t + 1)} \leq (1-\eps)pk \leq \frac{\rho^*(H)}{t + 1}\]
with probability at least $1 - n^{-4}$.
Therefore, when $\cAMPC$ invokes \cref{alg:mpc:DS}, we either attain a subgraph $S$ with density in $G$ at least $\frac{(1-\eps)\rho^*(G)}{2(1+\eps)^4(t + 1)}$ or $|S| \leq \frac{n}{(1+\eps/2)^L}$ while containing the densest subgraph of $H$, using \cref{lemma:DS-converge}.
We claim that one of the $\frac{4\log_{1+\eps} n}{L}$ invocations of \cref{alg:mpc:DS} will return a subgraph with the density guarantee.
We assume that \cref{alg:mpc:DS} does not return such a subgraph for the sake of contradiction.
Then, the size of the vertex set after the $\frac{4\log_{1+\eps} n}{L}$ invocations will be
\[n\cdot \rb{\frac{1}{(1+\eps/4)^L}}^{\frac{4\log_{1+\eps} n}{L}} < 1\]
which is not possible since the vertex set must contain the densest subgraph.
Therefore, given the existence of the parameter $k$ above, $\cAMPC$ will return a subgraph with density at least $\frac{(1-\eps)\rho^*(G)}{2(1+\eps)^4(t + 1)}$, which can be treated as a $(2+\eps)(t + 1)$-approximate densest subgraph by using a different $\eps$ when invoking \cref{alg:mpc:DS}.
\end{proof}

\section{$k$-core decomposition} \label{sec:k-core}

In this section, we apply our exponentiate and prune algorithm to the $k$-core decomposition problem.
See \cref{sec:basic-definition} for basic definitions and properties of the $k$-core decomposition problem.

\subsection{Base algorithm}
Our algorithm simulates the peeling algorithm from \cite{ghaffari2019improved}, presented in \cref{alg:kcore-peeling}.
The algorithm runs $\lceil \log_{1+\eps} n \rceil$ iterations.
Each iteration peels (removes) all vertices whose current degree is below $(2+2\eps)k$.
We call $k$ the \emph{degree parameter} and $(2+\eps)k$ the \emph{degree threshold}.
The algorithm produces a partial layer assignment $\ell_G$, where $\ell_G(v)$ is the iteration number in which $v$ is peeled, and $\ell_G(v) = \infty$ if $v$ is never peeled.

\begin{algorithm}
\caption{\kcorePeeling: Peeling algorithm for approximating coreness}
\label{alg:kcore-peeling}
\begin{algorithmic}[1]
\medskip
\Statex \textbf{Input:} graph $G$, parameters $\eps$, degree parameter $k$
\Statex \textbf{Output:} a partial layer assignment $\ell_G$
\medskip
\hrule
\medskip
\State $\ell(v) \gets \infty$ for all $v\in V(G)$
\For{$i = 1$ to $\lceil \log_{1+\eps} n \rceil$}
    \State $A\gets$ all vertices $v\in V(G)$ with $\deg_G(v) < (2+2\eps)k$
    \For{each $v\in A$}
        $\ell(v)\gets i$
    \EndFor
    \State Remove vertices in $A$ from $G$ \Comment{peeling step}
\EndFor
\State\Return $\ell_G$ and the labels of all vertices
\end{algorithmic}
\end{algorithm}

To analyze this algorithm, we need the following observation.

\begin{observation}[Folklore]
\label{observation:kcore}
Let $S \subseteq V(G)$ be the set of all vertices that have coreness at most $k$. Then, the number of edges incident to $S$ is at most $k|S|$.
\end{observation}

This observation implies the following key properties.

\begin{lemma}[\cite{ghaffari2019improved}]
\label{lemma:kcore-decrease}
Let $S$ be the set of all vertices that have coreness at most $k$. After peeling all vertices with degree below $(2+2\eps)k$, the number of remaining vertices in $S$ is at most $|S|/(1+\eps)$.
\end{lemma}
\begin{proof}
From \cref{observation:kcore}, we have that there are $k|S|$ edges incident to $S$.
Thus, at most
\[\frac{2k|S|}{(2+2\eps)k} = \frac{|S|}{1+\eps}\]
vertices in $S$ have degree $\geq (2+2\eps)k$.
Therefore, after peeling all vertices with degree below $(2+2\eps)k$, the number of remaining vertices in $S$ is at most $|S|/(1+\eps)$. 
\end{proof}

\begin{lemma}[\cite{ghaffari2019improved}] 
\label{lemma:kcore-peeling}
Let $G$ be a graph, $\eps\in (0, 1)$, $k \in \mathbb{N}$. Then, all vertices $v$ with coreness at most $k$ has $\ell_G(v) < \infty$, and all vertices $v$ with coreness at least $(2+2\eps)k$ has $\ell_G(v) = \infty$.
\end{lemma}
\begin{proof}
Using \cref{lemma:kcore-decrease}, we know that after $\log_{1+\eps} n$ iterations all vertices with coreness at most $k$ will be peeled.
Hence, all these vertices $v$ will be labeled with $\ell_G(v) < \infty$.
On the other hand, the vertices with coreness at least $(2+2\eps)k$ induce a subgraph of minimum degree $\geq (2+2\eps)k$, and hence they will never be peeled.
\end{proof}

The peeling algorithm can be used to compute approximate coreness as follows.
First, run \cref{alg:kcore-peeling} in parallel for parameters $\eps$ and $k_i = (1+\eps)^i$, where $i = 0, 1, \dots, \lceil \log_{1+\eps} n\rceil$.
If a vertex $v$ is not peeled with degree parameter $k_i$, \cref{lemma:kcore-peeling} ensures that $C(v) \geq k_i$.
Hence, we can compute a $(2+2\eps)$-approximate $k$-core decomposition $\tilde{C}$ as follows:
For each vertex $v$, assign $\tilde{C}(v)$ to be the largest $k_i$ such that $v$ is not peeled when the degree parameter is $k_i$.

\subsection{Preprocessing}
Fix a value $k$, and consider the subproblem of simulating \cref{alg:kcore-peeling} with parameter $k$.
The first step to apply our technique is to construct a sparser subgraph $H$ via random sampling, such that running the peeling algorithm on $G$ with degree parameter $k$ can be simulated by running the algorithm on $H$ with parameter $\Theta_\eps(\log n)$.
This step is formalized as follows.


Consider running \cref{alg:kcore-peeling} on a graph $G$ given some $\eps, k$.
Denote by $\tau = \lceil\log_{1+\eps}(n)\rceil$ the number of iterations.
For $1 \leq i \leq \tau + 1$, define
\[
    V_i(G, \eps, k) = \{v\in V(G) : \ell_G(v) \geq i\}.
\]
That is, $V_i(G, \eps, k)$ is the set of vertices that survives at least $i-1$ iterations.

The following lemma shows that running \cref{alg:kcore-peeling} on a sampled graph will result in all vertices with coreness $<k$ being peeled and all vertices with coreness at least $(2+\Theta(\eps))k$ not being peeled.

\begin{lemma}
\label{lemma:kcore-sampling}
Let $G$ be a graph, $k\in\mathbb{N}$, $\eps\in(0, 1/6)$. Let $p = \min\rb{1, \frac{20\log n}{\eps^2 k}}$. Let $H$ be a graph obtained from sampling each edge of $G$ independently with probability $p$. Then, the following holds with high probability:
\begin{itemize}
    \item (A) For all $1 \leq i \leq \tau + 1$, $V_i(H, \eps, (1+\eps)pk) \subseteq V_i(G, \eps, k)$.
    \item (B) The following holds for all $k' \geq k$.
    Let $K$ be the vertex set of the $k'$-core of $G$.
    In the subgraph of $H$ induced by $K$, the degree of each vertex is at least $pk'(1-\eps)$.
\end{itemize}
\end{lemma}
\begin{proof}
If $p = 1$, then $H = G$, and it is straightforward to verify the lemma.
Therefore, we assume $p = \frac{20\log n}{\eps^2 k} < 1$.

\begin{itemize}

\item \textbf{Proof of (A).} For all $1 \leq i \leq \tau + 1$, let $G_i$ be the subgraph of $G$ induced by $V_i(G, \eps, k)$, and let $H_i$ be subgraph of $H$ induced by $V_i(H, \eps, (1+\eps)pk)$.
We prove (A) by induction on $i$.
\begin{itemize}
    \item\textbf{Base case:} For $i = 1$, $V_i(H,\eps,(1+\eps)pk) = V_i(G,\eps,k) = V(G)$.
    \item\textbf{Inductive step:} Consider $i > 1$ and suppose that $V_{i - 1}(H, \eps, (1+\eps)pk) \subseteq V_{i - 1}(G, \eps, k)$.
    Let $T = V_{i - 1}(G, \eps, k) \backslash V_i(G, \eps, k)$ be the set of vertices peeled at iteration $i$ in $G$.
    Then, we want to show that $V_i(H, \eps, (1+\eps)pk) \cap T = \emptyset$.
    
    For all $v \in T$, we know that $\deg_{G_{i-1}}(v) < (2+2\eps)k$ since it is peeled.
    Now, we want to analyze $\deg_{H_{i-1}}(v)$. 
    We have that
    \[\E{\deg_{H_{i-1}}(v)} < (2+2\eps)pk\]
    since $V_{i - 1}(H, \eps, (1+\eps)pk) \subseteq V_{i - 1}(G, \eps, k)$.
    Therefore, using the Chernoff bound (\cref{lemma:chernoff}), we have
    \[\prob{\deg_{H_{i-1}}(v) \geq (2 + 4\eps)pk} \leq \exp\rb{-\eps^2(2 + 4\eps)pk/3} \leq n^{-6}.\]
    Taking the union bound over all $v\in T$, we see that $\deg_{H_{i-1}}(v) < (2 + 4\eps)pk$ with probability at least $1 - n^{-5}$.
    As a result, these vertices will be peeled and $V_i(H, \eps, (1+\eps)pk) \cap T = \emptyset$, implying that $V_i(H, \eps, (1+\eps)pk) \subseteq V_i(G, \eps, k)$.
    This ends our induction.
\end{itemize}

\noindent Taking the union bound over all $1 \leq i \leq \log_{1+\eps} n + 1$, we have that (A) holds with probability at least $1 - n^{-4}$.

\item \textbf{Proof of (B).}
Consider an integer $k' \geq k$.
Let $K$ be the set of vertices with coreness at least $k'$ in $G$.
Then, let $G'$ be the induced subgraph of $K$ in $G$, and let $H'$ be the induced subgraph of $K$ in $H$.

Since $G'$ is the $k'$-core, $\deg_{G'}(v) \geq k'$ for all $v\in K$.
Hence,
\[\E{\deg_{H'}(v)} \geq pk' \text{ for all } v \in K.\]
Using the Chernoff bound, we have that
\[\prob{\left|\deg_{H'}(v) - \E{\deg_{H'}(v)}\right| \geq \eps\E{\deg_{H'}(v)}} \leq 2\exp\rb{-\eps^2 p k' / 3} \leq n^{-6}.\]
Taking the union bound over all $v\in K$, we see that $\deg_{H'}(v) \geq (1-\eps)pk'$ with probability at least $1 - n^{-5}$.
\end{itemize}

\end{proof}

\subsection{The algorithm}

\subsubsection{Algorithm description}
Our algorithm for $k$-core is presented in \cref{alg:kcore-simulation}.
In the following, we refer to \cref{alg:kcore-simulation} as the \emph{simulated algorithm} and \cref{alg:kcore-peeling} as the \emph{base algorithm}.
The simulated algorithm first chooses parameters $L, B, a$ based on the input parameters $t, k, \eps$.
The algorithm is iterative, where each iteration attempts to simulate $L$ iterations of the base algorithm.
For technical reasons, the algorithm requires $10\tau / L$ iterations simulate all $\tau$ iterations of the base algorithm.


\begin{algorithm}
\caption{\kcoreSimulation: A simulation of \cref{alg:kcore-peeling}}
\label{alg:kcore-simulation}
\begin{algorithmic}[1]
\medskip
\Statex \textbf{Input:} graph $G$, parameters $t \geq 0, \eps \in (0, 1/6)$, $k \leq \frac{40\log n}{\eps^2}$
\Statex \textbf{Output:} a subset of vertices of $G$
\medskip
\hrule
\State $L \gets $ smallest power of two that is at least $\frac{\delta}{40 \lg k} \cdot (\lg n)^{^{(t+1)/(t+2)}}$
\State $B = n^{\delta/2}$, $s = \lceil \lg L / (t + 1) \rceil$, $\kappa = (2+\eps)k$, and $a = (t+1) \cdot \kappa$
\For{$i = 1$ to $\lceil 10\tau/L \rceil$}
    \State invoke $\Expo(G, \kappa, s, t, B)$ to obtain $(T_v, \map_v)$ for each $v \in V(G)$ \label{line:kcore-expo} 
    \State $\ell_v \gets \PartialLayerTree(T_v, a, L)$ for $v \in V(G)$ \label{line:kcore-local-assignment} 
    \State compute the label of each vertex $v$ as
    \[
        \ell(v) = \min\{\ell_u(x) \mid u \in V(G), x \in V(T_u), \text{ and } \map_u(x) = v.\}
    \]
    \State remove all vertices with $\ell(v) \neq \infty$. \Comment{peeling step} \label{line:kcore-peel}
\EndFor
\State\Return the set of remaining vertices
\end{algorithmic}
\end{algorithm}

Each iteration of the simulated algorithm first calls \Expo to compute pruned tree views of depth $2^{s(t+1)}$ for each vertex, where the parameters are chosen so that $2^{s(t+1)} \geq 2^{\lg L} \geq L$.
Next, \PartialLayerTree is invoked to label the nodes of each tree.
Then, line~\ref*{line:kcore-peel} peels all vertices that receive a finite label in any of the tree.

The properties of the algorithm are summarized as the following lemma.

\begin{lemma}
\label{lemma:k-core-log-case}
Given a graph $G$, parameters $t \geq 0, \eps \in (0, 1)$, and $k \leq \frac{40 \log n}{\eps^2}$, \cref{alg:kcore-simulation} computes a set of vertices $K$ such that 
\begin{itemize}
    \item all vertices with coreness $< k$ are not in $K$, and
    \item all vertices with coreness $\geq (2+\eps)(t+1)k$ are in $K$.
\end{itemize}
The algorithm requires $O\rb{\log^{1/(t+2)} n \cdot \poly\log\log (n)}$ rounds, $O(n^\delta)$ local space, and $O(n^{1+\delta} + m)$ global space, where $\delta \in (0,1)$ is any fixed constant.
\end{lemma}

The proof of \cref{lemma:k-core-log-case} is deferred to \cref{sec:kcore-analysis}.
We first show that the lemma can be combined with our random sampling to compute an approximate $k$-core decomposition. 

\begin{theorem}
\label{theorem:k-core}
There is an MPC algorithm that, given a graph $G$, parameters $t \geq 0$ and $\eps \in (0, 1/6)$, computes an approximate coreness value $\tilde{C}(v)$ for each vertex $v$ such that the following holds with probability $1 - n^{-5}$.
\[
    \tilde{C}(v) \leq C(v) \leq (2+\eps)(t+1)\tilde{C}(v) \text{ for every } v \in V(G).
\]
The algorithm requires $O\rb{\log^{1/(t+2)} n \cdot \poly\log\log (n)}$ rounds, $O(n^\delta)$ local space, and $\tO(n^{1+\delta} + m)$ global space, where $\delta \in (0,1)$ is any fixed constant.
\end{theorem}
\begin{proof}
Let $\eps_0 = \eps_0 / 30$.
For each $k_i = (1+\eps_0)^i$, $i \in [0, \log_{1+\eps_0} n]$, we perform the following.
First, obtain a subgraph $H_i$ by sampling each edge independently with probability $p_i = \min\rb{1, \frac{20 \log n}{\eps_0^2 k_i}}$.
If $k_i \geq \frac{20 \log n}{\eps_0^2}$, run \cref{alg:kcore-simulation} on $H_i$ with degree parameter $(1+\eps_0)p_ik_i = (1+\eps_0) \frac{20 \log n}{\eps_0^2}$ on $H_i$; otherwise, simply run \cref{alg:kcore-simulation} with degree parameter $k_i$ on $G$.
By \cref{lemma:kcore-sampling,lemma:k-core-log-case}, all vertices with coreness $< k_i$ in $G$ will be peeled, and all vertices with coreness at least $\frac{(2+\eps_0)}{(1- \eps_0)}(t+1)k \leq (2+7\eps_0)(t+1)k$ in $G$ will not be peeled.
Hence, for each $v$, we assign $\tilde{C}(v)$ as the largest $k_i$ such that $v$ is not peeled in the run for $k_i$.
It is not hard to verify that $\tilde{C}$ satisfies the property in the theorem.

Since the peeling for all $k_i$ can be done in parallel, the round and space complexities follow from \cref{theorem:missing}.
This completes the proof.
\end{proof}

\subsubsection{Overview of the analysis} \label{sec:kcore-analysis}
In the following, we present an overview of the proof of \cref{lemma:k-core-log-case} and introduce key technical lemmas.
Let $t \geq 0, \eps \in (0, 1)$ be approximation parameters and let $k$ be the degree parameter.
Assume that $k \leq \frac{40\log n}{\eps^2} = O_\eps(\log n)$.

Consider an execution of $\kcorePeeling(G, \eps, k)$.
Denote by $\ell_G$ the output labeling.
Let $\kappa = (2 + \eps)k$ be the degree threshold in the base algorithm.
It follows that $\ell_G$ has out-degree at most $\kappa$.

\begin{definition}[Layer sets]
For $i \in [1, \tau] \cup \{\infty\}$, define the \emph{$i$-th layer set} as $A_i \defeq \{v \mid v \in V(G) \text{ and } \ell_G(v) = i\}$.
Let $A_{\leq i} = \bigcup_{j \leq i} A_i$.
Define $A_{<i}$ and $A_{>i}$ similarly.    
\end{definition}

By \cref{lemma:kcore-decrease}, $A_{\leq \tau}$ contains the set of all vertices with coreness $\leq k$.
Hence, our main objective is to design a simulation of \kcorePeeling
that satisfy the following:
\begin{enumerate}
    \item All vertices in $A_{\leq \tau}$ are peeled. That is, the simulated algorithm peels at least as aggressively as the base algorithm.
    \item All vertices with coreness at least $(2+\eps)(t+1)k$ are not peeled. That is, the simulated algorithm at least preserves $(2+\eps)(t+1)k$-core.
\end{enumerate}

The second property is formulated below, and is proven in \cref{sec:kcore-preservation}.

\begin{lemma} \label{lemma:kcore-preservation}
    No vertex in $(2+\eps) \cdot (t+1)\cdot k$-core is removed during the execution of \cref{alg:kcore-simulation}.
\end{lemma}

We proceed to outline the proof for the first property.
Similar to edge orientation, we can show that all but $O(|A_{\leq L}| / 2^L)$ of the vertices in $A_{\leq L}$ can locally simulate $L$ steps of the peeling algorithm.
This implies that the number of vertices is reduced to $|A_{>L}| + O(|A_{\leq L}| / 2^L)$ after one iteration of \cref{alg:kcore-simulation}.
In edge orientation, such an upper bound already implies that the size of $V(G)$ has decreased by a $2^{\Theta(L)}$ factor.
This is not the case in the $k$-core setting, however, as $A_{>L}$ can be almost as large as the whole vertex set.

To overcome this challenge, we develop a strengthened argument, showing that when $|A_{\leq L}|$ is small, many vertices of $A_{>L}$ can also locally simulate the peeling process, even though we only compute pruned tree views of depth $L$.
To formalize this idea, we first divide layer sets into \emph{batches} of $L$ layers.

\begin{definition}[Layer batches]
For $j \in [1, \lceil \tau / L \rceil]$, define the \emph{$j$-th batch} as $\bigcup_{b = (j-1)\cdot L + 1}^{j \cdot L} A_b$.
Let $U_{\leq j} = \bigcup_{b \leq j} U_b$.
Define $U_{<j}$ and $U_{>j}$ similarly.
Denote by $p = \tau/L$ the number of batches.
For $i \in \{1, 2, \dots, 10\tau/L\}$, denote by $U_i^j$ the subset of vertices in $U_j$ that have not been peeled at the beginning of iteration $i$ in the simulated algorithm.
\end{definition}

By an adaptation of the edge orientation argument, we can show that $|U_1^i| \leq |U_1| \cdot \kappa^{10i \cdot L}$.
Hence, all of the first batch will be peeled in $O(\tau / L)$ iterations.
The size of the second batch is harder to upper bound, as the graph may contain several strictly increasing paths of length $\approx 2L$ ending at $U_2$;
hence, as long as not all vertices of $U_1$ are peeled, many vertices $U_2$ cannot simulate peeling locally by computing only pruned tree views of depth $L$.
However, we show that the number of such vertices can be bounded in terms of $|U_1|$.
Intuitively, if a vertex in $U_2$ is not the destination of any strictly increasing path of length $> L$, then it should be possible to simulate peeling on it with its pruned tree view of depth $L$.
This idea is formalized and generalized to higher batches as follows.

\begin{definition}[Critical vertex] \label{def:kcore-critical}
Consider a fixed iteration $i$.
A vertex $v \in G_i$ is critical with respect to $i$ if it satisfies both of the following.
\begin{itemize}
    \item (X1) No vertex in $U_{< j}^i$ can reach $v$ via a strictly increasing path (with respect to $\ell_G$) in $G_i$.
    \item (X2) $\NumPathsIn_{G_i,\ell_G,L} \leq B^{1/2^s}$.
\end{itemize}
\end{definition}

\begin{lemma} \label{lemma:kcore-critical}
Consider an iteration $i$ of \cref{alg:kcore-simulation}.
Let $v \in G_i$ be a critical vertex with respect to $i$.
Then, \cref{alg:kcore-simulation} must peel $v$ in iteration $i$.
\end{lemma}

The proof of \cref{lemma:kcore-critical} is presented in \cref{sec:kcore-critical}.
Equipped with this lemma, the analysis of round complexity reduces to upper bounding the number of critical vertices in each iteration.
We provide an upper bound in the following lemma.

\begin{lemma} \label{lemma:kcore-vertex-drop}
After $i^* \leq 10\tau/L$ iterations of \cref{alg:kcore-simulation}, we have $|U_j^{i^*}| = 0$ for all $j \in \{1, 2, \dots, p\}$.
\end{lemma}

See \cref{sec:kcore-vertex-drop} for a proof of this lemma.
An outline of the proof is as follows.
Fix an iteration $i$ and consider two batches $U_{j}^i, U_{j'}^i$, where $j \leq j'$.
By a counting argument, a vertex $v \in U_{j'}^i$ can reach $U_j^i$ by at most $\kappa^{(j-j'+1)L}$ increasing paths.
Hence, for a fixed $j$, the number of vertices not satisfying (X1) is roughly $\sum_{j' < j} |U_{j'}^i| \cdot \kappa^{(j - j' + 1)L}$.
The analysis for (X2) is similar.
(See \cref{sec:kcore-vertex-drop} for details.)

Using this observation, we upper bound $|U_j^i|$ by a recurrence
\[
    f(i+1, j) \leq f(i,j) / q^{9L} + \sum_{j' < j} f(i,j') \cdot q^{j-j'+1} \text{ for } i \ge 0, j \in [1, p].
\]
The recurrence can be upper bounded via a potential argument, showing that all batches are peeled in $10\tau/L$ iterations.

We conclude this overview with the proof of \cref{lemma:k-core-log-case}, assuming that \cref{lemma:kcore-preservation,lemma:kcore-vertex-drop} are true.
\begin{proof}[Proof of \cref{lemma:k-core-log-case} via \cref{lemma:kcore-preservation,lemma:kcore-vertex-drop}]
By \cref{lemma:kcore-peeling}, all vertices with coreness $< k$ in $G$ are contained in $A_1, A_2, \dots, A_\tau$.
\cref{lemma:kcore-vertex-drop} ensures that our simulation removes all of them in $O(\tau / L)$ iterations.
In addition, \cref{lemma:kcore-preservation} ensures that no vertex in $(2+\eps)(t+1)$-core is peeled.
Hence, the property in the lemma is satisfied.

The overall time and space complexities are determined by the invocations of \Expo. By \cref{theorem:missing}, each invocation takes $O(s \cdot t) = O(\log\log n)$ rounds, and there are $O(\tau / L) = O_\eps(\log^{1/(t+2)} n \cdot \log \log n)$ invocations.
Therefore, the round complexity is $O_\eps(\log^{1/(t+2)} n \cdot \log \log^2 n)$.
The local space is $O(B \cdot 2^{2^s}) = \tilde{O}(n^\delta)$, and the global space is $O(m) + nB^{1+o(1)} = \tilde{O}(n^{1+\delta} + m)$.
\end{proof}

\subsubsection{Proof of \cref{lemma:kcore-critical}} \label{sec:kcore-critical}

We first establish several properties of critical vertices.

\begin{lemma} \label{lemma:kcore-critical-transitivity}
    Consider an iteration $i$ and a critical vertex $v \in V(G_i)$.
    Let $u \in V(G_i)$ be a vertex that can reach $v$ via a strictly increasing path.
    Then, $u$ and $v$ are in the same batch.
    In addition, $u$ is also critical.
\end{lemma}
\begin{proof}
Let $U_j^i$ be the batch containing $v$.
Consider a vertex $u \in V(G_i)$ that can reach $v$ via strictly increasing path $P$.
The proof of this lemma consists of the following three parts.
\begin{itemize}
    \item \textbf{Part 1: $u$ and $v$ are in the same batch:}
    Clearly, $\ell_G(u) < \ell_G(v)$, and hence $u \in U_{\leq j}^i.$
    By (X1) of \cref{def:kcore-critical}, $u$ cannot be in $U_{< j}^i$.
    Hence, $u \in U_j^i$.

    \item \textbf{Part 2: $u$ satisfies (X1):}
    Suppose, by contradiction, that $u$ does not satisfy (X1).
    Since $u$ and $v$ are both in batch $U_{j}^i$, there exists a vertex $w \in U_{<j}^i$ that can reach $u$ via a strictly increasing path $P'$. Clearly, the concatenation of $P'$ and $P$ is a strictly increasing path that reaches $v$.
    This contradicts the assumption that $v$ satisfies (X1). Hence, $u$ must satisfy (X1).
        
    \item \textbf{Part 3: $u$ satisfies (X2):} 
    Let $\cP_u$ and $\cP_v$ be, respectively, the set of all strictly increasing paths ending at $u$ and $v$.
    Since there is a strictly increasing path from $u$ to $v$, any path in $\cP_u$ can be extended to a path in $\cP_v$.
    Hence, $|\cP_v| \geq |\cP_u|$.
    In Part 1 of our proof, we already proved that if a vertex $u'$ can reach $v$ via a strictly increasing path, then $u' \in U_j^i$.
    Therefore, the starting vertex of any path in $\cP_v$ is in $U_j^i$.
    Thus, any path in $\cP_v$ has length at most $L$, which shows that $|\cP_v| = \NumPathsIn_{G_i, \ell_G, L}(v)$.
    Consequently, we have
    \[
        \NumPathsIn_{G_i, \ell_G, L}(u) \leq |\cP_u| \leq |\cP_v| = \NumPathsIn_{G_i, \ell_G, L}(v).
    \]
    By our assumption, $\NumPathsIn_{G_i, \ell_G, L}(v) \leq B^{1/2^s}$, and hence $\NumPathsIn_{G_i, \ell_G, L}(u) \leq B^{1/2^s}$ as well.
    That is, $u$ satisfies (X2).
    This completes the proof of the lemma.
\end{itemize}

\end{proof}

\begin{lemma} \label{lemma:kcore-monotone-node}
    Let $v$ be a critical vertex with respect to $i$.
    Let $T_v$, $\map_v$ be the rooted tree and mapping computed in iteration $i$.
    Every monotonically reachable node $x \in T_v$ satisfies all of the following.
    \begin{enumerate}
        \item $\map_v(x)$ and $v$ are in the same batch.
        \item $\map_v(x)$ is active.
        \item $|\Missing_{G_i,T_v,\map_v}(x)| \leq t \cdot \kappa$.
    \end{enumerate}
\end{lemma}
\begin{proof}
    Let $r$ be the root of $T_v$.
    By definition of monotonically reachability, the path from $x$ to $r$ maps to a strictly increasing path $P$ from $\map_v(x)$ to $\map_v(r) = v$.
    By \cref{lemma:kcore-critical-transitivity}, we have $\map_v(x)$ is critical, and it is in the same batch as $v$.
    This proves the first property.
    
    Since $x$ satisfies (X2) of \cref{def:kcore-critical}, $\NumPathsIn_{G_i,\ell_G,L}(\map_v(x)) \leq B^{1/2^s}$.
    Hence, $\map_v(x)$ must be active.
    This proves the second property.
    The strictly increasing path has length at most $\ell_G(v) - \ell_G(\map_v(x)) + 1 \leq L$.
    This shows that the depth of $x$ is at most $L - 1$.
    The third property is then ensured by (P1) of \cref{theorem:missing}.
    This completes the proof.
\end{proof}


\begin{proof}[Proof of \cref{lemma:kcore-critical}]
    Consider a critical vertex $v \in G_i$.
    Let $T_v, \map_v$ be the rooted tree and mapping for $v$, and let $r$ be the root of $T_v$.
    Let $\ell_v$ be the partial layer assignment computed in iteration $i$.
    We claim that all monotonically reachable nodes $x \in V(T_v)$ satisfy $\ell_v(x) \leq L - (\ell_G(\map_v(r)) - \ell_G(\map_v(x)))$.
    


    We prove the claim by induction on $d_x \defeq \ell_G(\map_v(r)) - \ell_G(\map_v(x))$.
    Consider a monotonically reachable node $x$.
    By \cref{lemma:kcore-monotone-node}, $\map_v(x)$ is in the same batch as $v$.
    Hence, $d_x \in [0, L - 1]$.
    The basis of our induction is thus $d_x = L - 1$.
    \begin{itemize}
        \item \textbf{Basis: $d_x = L - 1.$}
        First, \cref{lemma:kcore-monotone-node} ensures that
        \[
            |\Missing_{G_i,T_v,\map_v}(x)| \leq t \cdot \kappa.
        \]
        We proceed to upper bound $|\children_{T_v}(x)|$.
        Consider a child $c$ of $x$.
        If $\ell_G(\map_v(c)) < \ell_G(\map_v(x))$, then $c$ is also monotonically reachable, and in addition $\ell_G(\map_v(c)) < \ell_G(\map_v(x)) = \ell_G(v) - (L-1)$.
        This shows that $\map_v(c)$ and $v$ are in different batches, contradicting \cref{lemma:kcore-monotone-node}.
        Therefore, any child of $x$ has $\ell_G(\map_v(c)) \geq \ell_G(\map_v(x))$.
        Since $\ell_G$ has out-degree $< \kappa$, this implies that $|\children_{T_v}(x)| \leq \kappa$.
        Consequently,
        \[
            |\children_{T_v}(x)| + |\Missing_{G_i,T_v,\map_v}(x)| < \kappa + t \cdot \kappa \leq a,
        \]
        and \PartialLayerTree must assign $\ell_v(x) = 1$. This proves the basis.
        
        \item \textbf{Inductive case:} Assume that the claim holds for all nodes $y \in V(T_v)$ with $d_y \in [d', L-1]$ for some $d' \leq L-1$.
        Consider a monotonically reachable node $x$ with $d_x = d' - 1$.
        Since $x$ is monotonically reachable, \cref{lemma:kcore-monotone-node} ensures that $|\Missing_{G_i,T_v,\map_v}(x)| \leq t \cdot \kappa$.
        
        We upper bound the number of children of $x$ as follows.
        Consider a child $c$ of $x$.
        Again, if $d_c > L - 1$, then $c$ is a monotonically reachable node that is not in the same batch as $v$, contradicting \cref{lemma:kcore-monotone-node}.
        If $d_c \in [d', L - 1]$, then by the induction hypothesis, \PartialLayerTree removes $c$ in one of iterations $1, 2, \dots, L - d'$.
        Hence, at the beginning of iteration $L - d' + 1$, every remaining child $c$ has $d_c \geq d'$, or equivalently, $\ell_G(\map_v(c)) \geq \ell_G(\map_v(x))$.
        Therefore, the number of children of $x$ is at most $\kappa$.
        Consequently, in iteration $d_x$
        \[
            |\children_{T_v}(x)| + |\Missing_{G_i,T_v,\map_v}(x)| < \kappa + t \cdot \kappa \leq a,
        \]
        and \PartialLayerTree must assign $\ell_v(x) = d_x$. This proves the inductive case.
    \end{itemize}

    By the claim above, every monotonically reachable node in $T_v$ receives a finite label in \PartialLayerTree.
    Since the root is also monotonically reachable, it also receives a finite label.
    Hence $v$ must be peeled in iteration $i$ of \cref{alg:kcore-simulation}.
    This completes the proof.
\end{proof}

\subsubsection{Proof of \cref{lemma:kcore-vertex-drop}} \label{sec:kcore-vertex-drop}

Consider a fixed iteration $i$ of the simulated algorithm.
\cref{lemma:kcore-critical} ensures that all critical vertices must be removed by the end of iteration $i$.
Hence, we prove this lemma by counting the number of non-critical vertices, i.e., the vertices that do not satisfy property (X1) or (X2) in \cref{def:kcore-critical}.

\begin{lemma} \label{lemma:kcore-path-bound-1}
Consider a fixed iteration $i$ of the simulated algorithm. For all $j \geq 0$, the number of vertices in $U_j^i$ that do not satisfy (X1) of \cref{def:kcore-critical} is at most
\[
    \sum_{j' < j} |U_{j'}^i| \cdot \kappa^{(j - j' + 1)L}.
\]
\end{lemma}
\begin{proof}
Fix integers $j' \leq j$.
Since each batch contains $L$ layers, the length of any strictly increasing path from $U_{j'}^i$ to $U_j^i$ is at most $L \cdot (j' - j + 1)$.
By \cref{claim:strictly-increasing-paths}, the number of such paths is at most
\[
    |U_{j'}| \cdot \sum_{b=1}^{(j'-j+1)L} \kappa^{b-1} \leq  |U_{j'}| \cdot \kappa^{(j'-j+1)L}.
\]
Hence, at most $|U_{j'}| \cdot \kappa^{(j'-j+1)L}$ vertices in $U_j^i$ can be reached by $U_{j'}^i$ via a strictly increasing path.

Summing over $j' < j$, the number of vertices that can be reached from $U_{<j}^i$ via a strictly increasing path is at most
\[
    \sum_{j' < j} |U_{j'}| \cdot \kappa^{(j'-j+1)L}.
\]
This proves the lemma.
\end{proof}

\begin{lemma} \label{lemma:kcore-path-bound-2}
Consider a fixed iteration $i$ of the simulated algorithm. For all $j \geq 0$, the number of vertices in $U_j^i$ that satisfy (X1) but not (X2) is at most
\[
    B^{-1/2^s} \cdot \kappa^L \cdot |U_j^i|.
\]
\end{lemma}
\begin{proof}
    Let $W$ be the set of vertices in $U_j^i$ that satisfy (X1).
    Every strictly increasing path ending at $U_j^i$ must come from $U_{\leq j}^i$.
    By (X1), any strictly increasing path ending at $W$ cannot come from $U_{< j}$.
    That is, all such paths are from $U_j^i$.
    
    Therefore,
    \begin{align*}
        &\ \sum_{v \in W} \NumPathsIn_{G_i, \ell_G, L}(v) \\ 
        \leq &\ \sum_{v \in U_{j}^i} \NumPathsOut_{G_i, \ell_G, L}(v) \\
        \leq &\ |U^i_{j}| \cdot \kappa^L,
    \end{align*}
    where the last inequality comes from \cref{claim:strictly-increasing-paths}.
    Therefore, the number of vertices with $\NumPathsIn_{G_i, \ell_G, L}(\cdot) \geq B^{1/2^s}$ is at most $\kappa^L \cdot |U_{j}^i| / B^{1/2^s}$.
    This completes the proof.
\end{proof}


We proceed to upper bound $|U_j^i|$ using \cref{lemma:kcore-path-bound-1,lemma:kcore-path-bound-2}.
Consider a fixed iteration $i$ and a batch $j$.
\cref{lemma:kcore-path-bound-1,lemma:kcore-path-bound-2} show that the number of non-critical vertices in $U_j^i$ is at most
\begin{align}
    &~\sum_{j' < j} |U_{j'}^i| \cdot \kappa^{(j-j'+1)L} + |U^i_{j}| \cdot B^{-1/2^s} \cdot \kappa^L \label{eqn:kcore-recursion1}
\end{align}

\noindent Our budget parameter $B$ is chosen so that $B^{1/2^s} \geq \kappa^{10L}$;
to see this, note that
\begin{align*}
\kappa^{10L \cdot 2^s} &\leq \exp_2(10\lg k \cdot L \cdot 2^{\lg(L) / (t+1) + 1}) \\ 
& = \exp_2(20\lg k \cdot L^{1+1/(t+1)}) \\
& \leq \exp_2(20\lg k \cdot \frac{\delta}{20 \lg k} \lg n) = n^{\delta} = B.
\end{align*}
Therefore, the term
\[
    |U^i_{j}| \cdot B^{-1/2^s}  \cdot \kappa^L \leq |U_j^i| / \kappa^{9L}.
\]

By \cref{lemma:kcore-critical}, all critical vertices in $U_j^i$ will be peeled, and hence $|U_j^{i+1}|$ is at most the number of non-critical vertices in $i$.
Denote $|U_j^i|$ by $f(i, j)$, and let $q = \kappa^L$.
\cref{eqn:kcore-recursion1} can be rephrased as the following recursion.

\begin{equation} \label{eqn:kcore-recursion2}
    f(i+1, j) \leq f(i,j) / q^{9} + \sum_{j' < j} f(i,j') \cdot q^{j-j'+1} \text{ for } i \ge 0, j \in [1, p].
\end{equation}

Initially, we have $\sum_{j\in[p]} f(0,j) = n$.
We claim that $\sum_{j\in[p]}f(i, j) \leq 0$ after $i = O(\log_q n + p)$ iterations.
To prove this, consider the potential function $M(i) \defeq \max_{1 \leq j \leq p} f(i, j) \cdot q^{-4j}$.
Note that $M$ is a real-valued function, while $f(\cdot, \cdot)$ is always an integer.
The following claim provides an upper bound for the potential; for completeness, we restate all parameters and function definitions in the claim.

\begin{claim} \label{claim:kcore-recursion-bound}
    Let $q$ be a positive integer, and $f: \mathbb{Z}_{\geq 0} \times [p] \rightarrow \mathbb{Z}_{\geq 0}$ be a function defined recursively by
    \[
        f(i+1, j) \leq f(i,j) / q^{9} + \sum_{j' < j} f(i,j') \cdot q^{j-j'+1} \text{ for } i \ge 0, j \in [1, p],
    \]
    where $f(0, j)$ can be any non-negative integer for $j \in [p]$.
    For $i \geq 0$, define the potential function $M(i)$ as $\max_{1\leq j \leq p} f(i, j) \cdot q^{-4j}$.
    Then, it holds that $M(i+1) \leq M(i) \cdot q^{-1}$ for $i \geq 0$.
\end{claim}
\begin{proof}
Fix some $j\in\{1,\ldots,p\}$.
Dividing \cref{eqn:kcore-recursion2} by $q^{4j}$ gives
\[
\frac{f(i+1, j)}{q^{4j}} \le \frac{f(i,j)}{q^{4j+9}} +
\sum_{j'<j} f(i,j') \cdot q^{j-j'+1-4j}.
\]
For the first term on the right hand side, since $\frac{f(i,j)}{q^{4j}}\le M(i)$,
we have
$\frac{f(i,j)}{q^{4j+9}} \le q^{-9}M(i)$.
For the summation term, write $d=j-j'$.
Then $d \ge 1$, and we have
\[
f(i,j') \cdot q^{j-j'+1-4j} = \frac{f(i,j')}{q^{4j'}}
q^{1-3d}.
\]
Therefore,
\[
\sum_{j'<j}
f(i,j')q^{j-j'+1-4j} \le M(i)\sum_{d\ge 1}q^{1-3d}.
\]
The geometric series satisfies $\sum_{d\ge 1}q^{1-3d} = \frac{q^{-2}}{1-q^{-3}}$. Since $q\ge 2$, we have
$q^{-9}+\frac{q^{-2}}{1-q^{-3}} \le q^{-1}$.
Combining the above bounds, we get $\frac{f(i+1,j)}{q^{4j}} \le q^{-1}M(i)$.
Since this holds for every $j$, taking the maximum over $j$ gives $M(i+1)\le q^{-1}M(i)$.
This completes the proof.
\end{proof}

\cref{claim:kcore-recursion-bound} implies that $M(i) \leq M(0) \cdot q^{-i}$.
Take $i^* = \lceil \log_q n\rceil + 4p + 1$, we have
\[
    M(i^*) < M(0) \cdot \frac{1}{n} \cdot q^{-4p} \leq q^{-4p}.
\]
Hence, for all $j \in \{1, 2, \dots, p\}$,
\[
    f(i^*, j) \leq M(i^*) \cdot q^{4j} < 1.
\]
That is, $|U_j^{i^*}| = 0$ for all $j$.
Since $i^* \leq \lceil \frac{\log n}{\log q} \rceil + 4\lceil\frac{\tau}{L}\rceil + 1 < \frac{10\tau}{L}$, \cref{lemma:kcore-vertex-drop} holds.

    

\subsubsection{Proof of \cref{lemma:kcore-preservation}} \label{sec:kcore-preservation}

Let $K$ be the $\kappa(t+1)$-core of $G$.
In the following, we prove that the simulated algorithm never removes a vertex in $K$.
Intuitively, this is because an iteration of \PartialLayerTree only labels a vertex if its current degree is at most $a < \kappa(t+1)$, and such a procedure should never remove any vertex in $\kappa(t+1)$-core.

We need the following folklore property, which can be used to identify the vertices not in $K$.

\begin{claim}[Folklore] \label{claim:kcore-elimination}
    Let $S$ be a vertex subset that contains $V(K)$.
    If a vertex $v$ has degree $< \kappa(t+1)$ in the subgraph induced by $S$, then $v \notin V(K)$.
\end{claim}

\begin{lemma} \label{lemma:k-core-preservation-sublemma}
Consider an iteration $i$ of the simulated algorithm, and assume that $K \subseteq V(G_i)$. Then, it holds that $K \subseteq V(G_{i+1})$.
\end{lemma}
\begin{proof}
    Let $v \in V(G_i)$ be a vertex in $G_i$.
    Consider the execution of \cref{line:kcore-local-assignment}, where \PartialLayerTree is invoked to give each node $x \in V(T_v)$ a label.
    We claim that all nodes $x$ that receive a finite label in this step satisfy $\map_v(x) \notin K$.

    Recall that \PartialLayerTree executes $L$ steps, where each step removes the set of nodes $x$ such that
    \[
        |\children_{T_v}(x)| + |\Missing_{G_i,T_v,\map_v}(x)| < \kappa(t+1).
    \]
    Let $X_j \subseteq V(T_v)$ be the set of nodes removed in step $i$.
    Let $Y_j = \{\map_v(x) \mid x \in X_j\}$.
    The set of vertices that receive finite labels is hence $\bigcup_{j \geq 1} Y_j$.
    Let $G_i(j) = G_i \setminus \bigcup_{j' < j} Y_{j'}$.
    Then, by the definition of the missing set, we have the following property.
    \begin{itemize}
        \item Let $j \in \{1, 2, \dots, L\}$ be an index. When a node $x \in X_j$ is removed, the degree of $\map_v(x)$ in $G_i(j)$ is smaller than $\kappa(t+1)$.
    \end{itemize}
    
    The above property and \cref{claim:kcore-elimination} directly imply that no vertex in $Y_1$ is in $K$.
    In addition, by an induction on $j$, we know that no vertex in $Y_j$ is in $K$ for $j = 1, 2, \dots, L$.
    This shows that no node $x \in V(T_v)$ with a finite label maps to a vertex in $K$.

    Since the argument holds for all $v \in V(G_i)$, any vertex removed in iteration $i$ is not in $K$.
    Hence, the lemma holds.
\end{proof}

By repeatedly applying \cref{lemma:k-core-preservation-sublemma}, we know that each iteration of the simulated algorithm only removes vertices not in $K$.
Consequently, \cref{lemma:kcore-preservation} holds.

\section*{Acknowledgments}
AI Disclosure: The authors used ChatGPT 5.5 to prove an upper bound for the recurrence in \cref{eqn:kcore-recursion2}. The upper bound was first conjectured by the authors based on computations of small cases. The tool generated a proof, which can be verified by a straightforward induction argument. More specifically, the tool materially affected the proof of \cref{claim:kcore-recursion-bound}. The authors verified the correctness and originality of all content including references. 

\bibliographystyle{alpha}
\bibliography{ref.bib}

\appendix

\section{$\tO(L/\sqrt{\log n})$-round densest subgraph algorithm} \label{appendix:sqrt-DS}
In this section, we present the modified algorithm \dsSqrtPeeling from \cite{mitrovic2026new}.
The goal of \dsSqrtPeeling is to serve as a subroutine that can compute an approximate densest subgraph in $O(c\cdot \log\log n)$ rounds for some parameter $c$.
Normally, the algorithm would take $\tO(\sqrt{\log n})$ rounds, which is the guarantee that \cite{mitrovic2026new} proves, but the way our algorithm in \cref{sec:DS-alg} uses \dsSqrtPeeling as a subroutine allows a speedup dependent on $c$ when $c = o(\sqrt{\log n})$.
Specifically, depending on the parameter of $L$ used in \cref{sec:DS-alg}, we use $c = \Theta(L/\sqrt{\log n})$ which results in $\tO(L/\sqrt{\log n})$ total rounds.
We present \dsSqrtPeeling in \cref{alg:ds-sqrt-peeling}.

\begin{algorithm}
\caption{(\dsSqrtPeeling) Finds approximate densest subgraph}
\label{alg:ds-sqrt-peeling}
\begin{algorithmic}[1]
\medskip
\Statex\textbf{Input:} graph $G$, $\eps \in (0,1)$, $\delta\in(0,1)$, $c, k\in\mathbb{N}$
\Statex\textbf{Output:} a subset of vertices of $G$
\medskip
\hrule
\medskip
\State $\ell(v) \gets \infty$ for all $v\in V(G)$
\State $i \gets 1$

\For{$j = 1$ to $c$}

\State $\alpha \gets (1+\eps)^{\sqrt{\log_{1+\eps} n}}$

\State Freeze all vertices in $V(G)$ of degree greater than $k\alpha$

\State Mark as frozen each edge with both endpoints frozen

\State $f \gets$ number of frozen vertices in $V(G)$

\If{$f \geq n/\alpha$}
    \Return $V(G)$
\EndIf

\For{$\sqrt{\delta\log_{1+\eps} n}/2$ steps}

    \If{$|V(G)| \leq (2+4\eps)f/\eps$}
        \textbf{break}
    \EndIf
    
    \State $A \gets$ all non-frozen vertices $v \in V(G)$ with $d_{G}(v) < k$

    \If{$|A| \leq \frac{\eps}{1+\eps}|V(G)| - f$}
        \Return $V(G)$
    \EndIf

    \For{each $v\in A$}
        $\ell(v)\gets i$
    \EndFor

    \State $i\gets i + 1$
    
    \State Remove vertices in $A$ from $G$
    
\EndFor

\EndFor
\end{algorithmic}
\end{algorithm}

For the analysis of \dsSqrtPeeling, we assume $k = O(\log n)$.
In other words, we assume we have already applied the sampling scheme described in \cref{sec:ds-sampling}.

\subsection{Simulation in sublinear MPC}

\dsSqrtPeeling uses the same framework from \cite{ghaffari2019sparsifying}, breaking up $O(\log n)$ total iterations of peeling from \dsPeeling into blocks of $\Theta(\sqrt{\log n})$ iterations.
These blocks are reflected in the inner for-loop of \dsSqrtPeeling.
We call each iteration of the outer for-loop a \textit{phase}.
One phase can be simulated by freezing high degree vertices and performing exponentiation on the low degree vertices.
Therefore, simulating each phase takes $O(\log \log n)$ rounds giving \dsSqrtPeeling a total round complexity of $O(c\cdot \log \log n)$ rounds.
For more details on how to simulate the algorithm in sublinear MPC using $O(n^\delta)$ local memory and $O(m + n^{1+\delta})$ total memory, we refer the reader to \cite{ghaffari2019sparsifying, mitrovic2026new}.
\begin{lemma}
\label{lemma:sqrt:memory-round}
Let $G$ be a graph, $\eps,\delta\in(0, 1)$, $c, k\in\mathbb{N}$. Then, $\dsSqrtPeeling(G, \eps, \delta, c, k)$ can be simulated by a MPC algorithm which requires $O(c\cdot\log\log n)$ rounds, $O(n^\delta)$ local space, and $O(n^{1+\delta} + m)$ global space.
\end{lemma}
\begin{proof}
For each phase, we do $O(\log \log n)$ iterations of graph exponentiation on every nonfrozen vertex using non-frozen edges. Since there are $c$ phases, this results in $O(c\cdot\log\log n)$ total MPC rounds. For more details on graph exponentiation, we refer the reader to \cite{ghaffari2019sparsifying, mitrovic2026new}.

These iterations of graph exponentiation allow a $\sqrt{\delta\log_{1+\eps} n}/2$-hop neighborhood of a nonfrozen vertex to be stored on a single machine. The machine then uses local space
\[\rb{k\alpha}^{\sqrt{\delta\log_{1+\eps} n}/2} \in O(n^\delta).\]
Additionally, since every vertex stores its neighborhood on a separate machine, this results in a global memory of $O(n^{1+\delta} + m)$.
\end{proof}

\subsection{Approximation guarantee}
We want to show that \dsSqrtPeeling outputs a subgraph with density at least $\frac{k}{2(1+\eps)}$.
As a result, if $k$ is chosen close to $\rho^*(G)$, we compute a $(2+\eps)$-approximate densest subgraph.
We observe that only lines $8$ and $12$ return a subgraph.
From these lines, we present the following lemma.
\begin{lemma}
\label{lemma:sqrt:approx}
Let $G$ be a graph, $\eps,\delta\in(0, 1)$, $c, k\in\mathbb{N}$. Then, if $\dsSqrtPeeling(G, \eps, \delta, c, k)$ returns a subgraph $S$, $S$ satisfies $\rho(S) \geq \frac{k}{2(1+\eps)}$. 
\end{lemma}
\begin{proof}
If $S$ is returned on line $8$ of \cref{alg:ds-sqrt-peeling}, then the number of vertices with degree greater than $k\alpha$ is at least $n/\alpha$.
This means that the graph $G$ at that moment has at least $nk/2$ edges.
Therefore, returning $S = V(G)$ would have density $\rho(S)\geq k/2$.

If $S$ is returned on line $12$ of \cref{alg:ds-sqrt-peeling}, we have that at least $|V(G)|/(1+\eps)$ vertices have degree at least $k$.
Therefore, returning $S = V(G)$ would have density
\[\rho(S) \geq \frac{\frac{|V(G)|}{2(1+\eps)}\cdot k}{|V(G)|} \geq \frac{k}{2(1+\eps)}.\]
\end{proof}

\subsection{Vertex set analysis}
Finally, we want to prove properties about the change in the vertex sets of the graph as well as the layers assigned to vertices by \dsSqrtPeeling.
We define vertex sets $V^{\text{sqrt}}_i$ as
\[V^{\text{sqrt}}_i = \{v \in V(G) : \ell_{\text{sqrt}}(v) \geq i\}\]
where $\ell_{\text{sqrt}}$ is the layer assignment produced by \dsSqrtPeeling.
So, $V^{\text{sqrt}}_\infty$ contains all vertices with layer $\infty$.
We want the sizes of these vertex sets to decrease by at least a constant factor as $i$ increases.
We present this in \cref{lemma:sqrt:vertex-sets-decrease}.
\begin{lemma}
\label{lemma:sqrt:vertex-sets-decrease}
Let $G$ be a graph, $\eps,\delta\in(0, 1)$, $c, k\in\mathbb{N}$. Let $V^{\text{sqrt}}_i$ be the vertex sets produced by $\dsSqrtPeeling(G, \eps, \delta, c, k)$.
Then, for all $i\geq 1$,
\[|V^{\text{sqrt}}_{i + 1}| \leq \frac{|V^{\text{sqrt}}_{i}|}{1+\eps/2} \text{ or }  V^{\text{sqrt}}_{i} = V^{\text{sqrt}}_\infty.\]
\end{lemma}
\begin{proof}
Let us assume that $V^{\text{sqrt}}_{i} \neq V^{\text{sqrt}}_\infty$.
Then, the algorithm must have done an iteration of peeling that removes vertices in $V^{\text{sqrt}}_{i}$ to produce the vertex set $V^{\text{sqrt}}_{i + 1}$.
The number of vertices that are removed from $V^{\text{sqrt}}_{i}$ is lower bounded by $\frac{\eps}{1+\eps}|V^{\text{sqrt}}_{i}| - f$ from line $12$ of \cref{alg:ds-sqrt-peeling}.
Additionally, from line $10$ of \cref{alg:ds-sqrt-peeling}, we can upper bound $f$ by $\frac{\eps|V^{\text{sqrt}}_{i}|}{2+4\eps}$.
Therefore, we have that
\begin{eqnarray*}
    |V^{\text{sqrt}}_{i + 1}| &\leq & |V^{\text{sqrt}}_{i}| - \rb{\frac{\eps}{1+\eps}|V^{\text{sqrt}}_{i}| - f}\\
    &=& \frac{|V^{\text{sqrt}}_{i}|}{1 + \eps} + f\\
    &\leq& \frac{|V^{\text{sqrt}}_{i}|}{1 + \eps} + \frac{\eps|V^{\text{sqrt}}_{i}|}{2+4\eps}\\
    &\leq& \frac{|V^{\text{sqrt}}_{i}|}{1+\eps/2}.
\end{eqnarray*}
This gives us our desired vertex set size reduction.
\end{proof}

Using \cref{lemma:sqrt:vertex-sets-decrease}, we now show that after one phase, the vertex set of the graph decreases by a significant factor.
This significant decrease is what allows us to shave a factor of $\Theta(\sqrt{\log n})$ from the round complexity when \dsSqrtPeeling is used in \cref{sec:DS-alg}.
We formalize the decrease in vertex set size in \cref{lemma:sqrt:vertex-set-decrease-outer}.
\begin{lemma}
\label{lemma:sqrt:vertex-set-decrease-outer}
Let $G$ be a graph, $\eps,\delta\in(0, 1)$, $c, k\in\mathbb{N}$. Let $V^{\text{sqrt}}_i$ be the vertex sets produced by $\dsSqrtPeeling(G, \eps, \delta, c, k)$. Let $a_0 = 1$ and $a_j > 1$ be an integer where $V^{\text{sqrt}}_{a_j}$ is the vertex set after one phase of \dsSqrtPeeling for $1 \leq j \leq c$.
Then, for all $0 \leq j \leq c - 1$ with $V^{\text{sqrt}}_{a_{j + 1}} \neq V^{\text{sqrt}}_\infty$, we have that
\[|V^{\text{sqrt}}_{a_{j + 1}}| \leq \frac{|V^{\text{sqrt}}_{a_j}|}{(1+\eps/2)^{\sqrt{\delta\log_{1+\eps}n}/2}}.\]
\end{lemma}
\begin{proof}
Let us consider any $0 \leq j \leq c - 1$ with $V^{\text{sqrt}}_{a_{j + 1}} \neq V^{\text{sqrt}}_\infty$.
Since $V^{\text{sqrt}}_{a_{j + 1}} \neq V^{\text{sqrt}}_\infty$, we know that \cref{alg:ds-sqrt-peeling} does not return on line $8$ or $12$.
Now, we have two situations: either the inner for-loop goes through all $\sqrt{\delta\log_{1+\eps}n}/2$ steps, or it breaks early on line $10$.

If the inner for-loop goes through all $\sqrt{\delta\log_{1+\eps}n}/2$ steps, then we have that $a_{j + 1} - a_j = \sqrt{\delta\log_{1+\eps}n}/2$.
This gives us
\[|V^{\text{sqrt}}_{a_{j + 1}}| \leq \frac{|V^{\text{sqrt}}_{a_j}|}{(1+\eps/2)^{a_{j + 1} - a_j}} = \frac{|V^{\text{sqrt}}_{a_j}|}{(1+\eps/2)^{\sqrt{\delta\log_{1+\eps}n}/2}}\]
using \cref{lemma:sqrt:vertex-sets-decrease}.

Otherwise, the inner for-loop breaks early on line $10$.
This means that $|V^{\text{sqrt}}_{a_{j + 1}}| \leq (2+4\eps)f/\eps$.
Combining this with $f < |V^{\text{sqrt}}_{a_j}|/\alpha$ from line $8$, we have that
\[|V^{\text{sqrt}}_{a_{j + 1}}| \leq (2+4\eps)f/\eps \leq \frac{(2+4\eps)|V^{\text{sqrt}}_{a_j}|}{\eps (1+\eps)^{\sqrt{\log_{1+\eps}n}}} \leq \frac{|V^{\text{sqrt}}_{a_j}|}{(1+\eps/2)^{\sqrt{\delta\log_{1+\eps}n}/2}}\]
for sufficiently large $n$.
Therefore, we have our desired upper bound.
\end{proof}

Finally, we want to confirm that \cref{alg:ds-sqrt-peeling} will never remove all the vertices of the graph for $k\leq \rho^*(G)$.
Specifically, we show that the vertices in the $k$-core of the graph will never be removed, and the $k$-core is nonempty for $k\leq \rho^*(G)$.

\begin{lemma}
\label{lemma:sqrt:nonempty-kcore}
Let $G$ be a graph, $\eps,\delta\in(0, 1)$, $c, k\in\mathbb{N}$ with $k\leq \rho^*(G)$. Let $K$ be a nonempty set of vertices in the $k$-core of $G$. Then, $\dsSqrtPeeling(G,\eps, \delta, c, k)$ never removes any vertex in $K$.
\end{lemma}
\begin{proof}
We know that the vertices in the $k$-core of $G$ all have degree at least $k$ within the induced subgraph of $K$.
Because \dsSqrtPeeling only removes vertices when their degree is less than $k$, none of the vertices in $K$ will be removed.
Additionally, $K$ is nonempty since the densest subgraph is a subset of the $\rho^*(G)$-core of $G$ and $k\leq \rho^*(G)$.
\end{proof}

We now have all the necessary properties of the vertex sets produced by \dsSqrtPeeling.
They are used to prove the approximation guarantees of our densest subgraph algorithm in \cref{sec:DS-alg}.

\end{document}